\documentclass[11pt]{article}
\usepackage[T1]{fontenc}
\usepackage{lmodern}
\usepackage{microtype}
\usepackage{amsmath,amssymb,amsthm,mathtools}
\usepackage{comment}
\usepackage[numbers,sort&compress]{natbib}
\usepackage{enumitem,array,needspace}
\usepackage{authblk}
\usepackage[colorlinks=true,
            linkcolor=black,
            citecolor=black,
            urlcolor=black]{hyperref}

\usepackage{orcidlink}
\usepackage{geometry}
\numberwithin{equation}{section}

\newtheorem{theorem}{Theorem}[section]
\newtheorem{proposition}[theorem]{Proposition}
\newtheorem{lemma}[theorem]{Lemma}
\newtheorem{corollary}[theorem]{Corollary}
\theoremstyle{definition}

\theoremstyle{remark}
\newtheorem{remark}[theorem]{Remark}

\newcommand{\one}{\mathbf 1}
\newcommand{\dd}{\,\mathrm d}
\newcommand{\scal}{\operatorname{sc}}
\newcommand{\az}{$\alpha$--$z$}

\newcolumntype{L}[1]{>{\raggedright\arraybackslash}p{#1}}

\newcommand{\defeq}{\coloneqq}
\DeclareMathOperator{\Tr}{Tr}
\DeclareMathOperator{\tr}{tr}
\newcommand{\cH}{\mathcal H}
\newcommand{\cN}{\mathcal N}
\newcommand{\e}{\mathrm e}
\usepackage{braket}
\hypersetup{pdftitle={Integral representations of quantum Renyi divergences on von Neumann algebras},pdfauthor={Tanay Kibe and Pratik Roy}}

\begin{document}

\title{Integral representations of quantum R\'enyi divergences\\
on von Neumann algebras}

\author[1]{Tanay Kibe\,\orcidlink{0000-0002-0571-1712} \thanks{tanay.kibe@ib.edu.ar}}
\affil[1]{Instituto Balseiro, Centro At\'omico Bariloche,
S.C. de Bariloche, 8400 R\'io~Negro, Argentina}

\author[2]{Pratik Roy\,\orcidlink{0000-0003-4327-3132} \thanks{\href{mailto:roy.pratik92@gmail.com}{roy.pratik92@gmail.com}}}
\affil[2]{Institute of Mathematics, University of Warsaw,
ul. Banacha 2, 02-097 Warsaw, Poland}

\date{}

\maketitle

\begin{abstract}
We prove an integral representation of the $\alpha$--$z$ R\'enyi divergences for normal states on arbitrary von Neumann algebras. Along each ray $z=c\alpha$, the divergence is a positive average of relative entropies evaluated in a canonical family of escort states. For orders $\alpha$ above one, we require the divergence to be finite, whereas for orders below one support inclusion suffices. An explicit boundary correction removes the lower-order support assumption. The proof uses a compatible weighted Haagerup $L^p$ scale and applies outside the data-processing region. For faithful density matrices, we express the derivative of the escort relative entropy with respect to the R\'enyi order in terms of a nonnegative fluctuation quantity. Differences of order-one derivatives of the divergence give quantum Fisher information and Wigner--Yanase skew information. For faithful bipartite density matrices, we also express optimized conditional $\alpha$--$z$ entropies at orders above one as averages of ordinary conditional entropies. Further applications concern coherent-state divergences in quantum field theory and the optimizing order for the strong converse exponent in hypothesis-testing.
\end{abstract}

\medskip

\newpage
\tableofcontents

\section{Introduction}
\label{sec:introduction}

Relative entropy is one of the principal tools for comparing states in
quantum statistical mechanics and quantum field theory. For density
matrices it is given by
\[
 D(\rho\Vert\sigma)=\Tr\rho(\log\rho-\log\sigma),
\]
with the usual support convention. Araki extended this definition to
normal states on arbitrary von Neumann algebras
\cite{Araki:1976entropy}. This algebraic formulation is essential in local
quantum field theory, where observable algebras are generally of Type III 
and states are not represented by density matrices. Modular
theory nevertheless permits explicit calculations involving relative entropy, 
leading to energy inequalities and bounds on modular evolution
\cite{Ceyhan:2018recovering,Hollands:2025qnec,Hollands:2025modular}.

Quantum R\'enyi divergences provide further measures of distinguishability.
Among them, the sandwiched R\'enyi divergence (SRD) has a central role in quantum
hypothesis testing \cite{Muller-Lennert:2013renyi,Wilde:2013converse}. SRD belongs
to the more general two-parameter $\alpha$--$z$ family
\cite{Audenaert:2013entropies}, defined for normal states on
arbitrary von Neumann algebras using Haagerup noncommutative $L^p$ spaces
\cite{Kato:2023spaces,Kato:2023divergence,Hiai:2024reversibility}; see also \cite{Berta:2016weighted,Jencova:2016entropies,Jencova:2017spaces} for SRD on general algebras.

Motivated by the role of relative entropy in quantum energy
inequalities \cite{Wall:2017lower}, 
we establish an exact relation between quantum R\'enyi
divergences and relative entropy. 
Among other things, such a relation makes the
methods and physical interpretations developed for relative entropy
available to the study of R\'enyi quantities.

On arbitrary von Neumann algebras, we prove 
an integral representation 
expressing $\alpha$--$z$ divergences as integrals of relative entropies
along a canonical family of escort states. The construction is particularly simple for faithful density matrices.
Fix $c=z/\alpha>0$ and define the escort state\footnote{The terminology follows the classical escort construction
$p_i\mapsto p_i^\beta/\sum_jp_j^\beta$, used in nonextensive statistical 
mechanics \cite[Chap.~9]{Beck:1993thermodynamics}.
Its information geometry was studied by Abe
\cite{Abe:2003geometry}. For commuting faithful density matrices,
our escort distribution is proportional to
$p_i^\beta s_i^{1-\beta}$, where $(s_i)$ is the reference distribution.
The usual escort construction is recovered when the reference
state is maximally mixed, $\sigma=\one/d$.} at R\'enyi order $\beta$ by
\begin{equation}
 \rho_\beta^{(c)}=
 \frac{\left(\sigma^{\frac{1-\beta}{2c\beta}}
             \rho^{1/c}
             \sigma^{\frac{1-\beta}{2c\beta}}\right)^{c\beta}}
 {\Tr\left(\sigma^{\frac{1-\beta}{2c\beta}}
             \rho^{1/c}
             \sigma^{\frac{1-\beta}{2c\beta}}\right)^{c\beta}},
 \qquad \beta>0,
 \label{eq:intro-escort}
\end{equation}
with $\rho_1^{(c)}=\rho$ and $\beta$ ranging from 1 to $\alpha$.
When $\rho$ and $\sigma$ commute, this construction
reduces to the classical tilted distributions associated with
R\'enyi divergence \cite{Beck:1993thermodynamics,vanErven:2012divergence}. Related
noncommutative constructions appear in the entropic functionals of
quantum statistical mechanics \cite{Jaksic:2011entropic}.
Our integral representation is
\begin{equation}
 D_{\alpha,z}(\rho\Vert\sigma)
 =\frac{\alpha}{\alpha-1}
  \int_1^\alpha
  \frac{D(\rho_\beta^{(c)}\Vert\sigma)}{\beta^2}\,\dd\beta,
 \qquad c=\frac z\alpha.
 \label{eq:intro-fd-formula}
\end{equation}
For $0<\alpha<1$, reversing the limits makes this a positive average. 
The sandwiched divergence corresponds to $c=1$. The integration variable here
is the R\'enyi order,  which differs from representations based on
quantum hockey-stick quantities, whose integration variable is a
testing threshold
\cite{Frenkel:2022integral,Hirche:2023integral,Liu:2025representations}.

The integral representation~\eqref{eq:intro-fd-formula} transfers the dependence on the
R\'enyi order to the escort state, while keeping both the
reference state and the relative-entropy functional fixed. Its
usefulness therefore rests on identifying the escorts or estimating
their relative entropies. When this is possible, relative-entropy
calculations and bounds can be integrated to obtain the corresponding
R\'enyi quantities. In particular, the representation provides a way
to borrow the modular-energy--entropy interpretation associated with 
the relative entropy into the study
of SRD and $\alpha$--$z$ divergences
\cite{Kibe:2026focusing,Kibe:2026escorts}. 

The finite-dimensional differential identity underlying
\eqref{eq:intro-fd-formula} has some antecedents. Bao, Moosa, and
Shehzad \cite{Bao:2019holographic} identified the refined 
sandwiched divergence with the relative
entropy of the sandwiched escort state. 
Integrating their identity gives the $c=1$ integral-representation for the matrix case.
The differentiation argument for general fixed rays also appears in
the proof of \cite[Lemma~10.5]{Rubboli:2024conditional}, where it is
used to establish parameter monotonicity. Our main contribution is
to prove the representation on arbitrary von Neumann algebras and
to determine the support and finiteness conditions under which it
holds.

More precisely, Theorem~\ref{thm:fixed-ray-integral} establishes the
representation for normal states $\psi,\omega$ on an arbitrary von
Neumann algebra. For $0<\alpha<1$ and $z>0$, we assume
$s(\psi)\leq s(\omega)$. For $\alpha>1$ and $z>0$, finiteness of
$D_{\alpha,z}(\psi\Vert\omega)$ at the specified values of $\alpha,z$ 
suffices. The theorem imposes no restriction to the data-processing region.
If support inclusion fails below order one, compression to the
support of the reference state produces an explicit boundary
correction
(Proposition~\ref{prop:general-fixed-ray-boundary}). This correction
is necessary even for classical probability distributions.

For matrices, the identity follows by differentiating a trace of a
positive power. The general case requires a different argument because
the factors can be unbounded and their $L^p$ exponents vary with the
order. We place these factors in a compatible family of weighted
Haagerup spaces. An exact difference quotient and the order-one limit of the
$\alpha$--$z$ divergence identify the derivative of a suitably
rescaled divergence with the escort relative entropy. 
Convexity then supplies the absolute continuity required
to integrate it. The interpolation theorem of Gu, Yin, and Zhang
\cite{Gu:2019interpolation} allows this argument to include exponents
below one, where the relevant Haagerup spaces are quasi-Banach. 

The relevance of the algebraic formulation is illustrated by coherent
states in quantum field theory. For the family of states considered here,
modular theory determines the relative entropy of the escorts, so
the order integral yields the sandwiched divergence in closed form.
For the $\alpha$--$z$ family it gives a modular-frequency integral
that also determines the precise finiteness threshold. These
examples show how relative entropy can be used to compute continuum
R\'enyi quantities and study their dependence on a cross section of a null plane. 
These calculations are developed in \cite{Kibe:2026escorts} and enter the
focusing and R\'enyi quantum null energy condition and ant-formula applications of
\cite{Kibe:2026focusing,Kibe:2026integer}. We give a brief overview of these applications in Section \ref{subsec:qft-applications}.

Beyond such evaluations, the representation gives a local description
of how distinguishability changes with the R\'enyi order.  For faithful density matrices, write
$R_c(\beta)=D(\rho_\beta^{(c)}\Vert\sigma)$ and
$L_\beta=\log\rho_\beta^{(c)}-\log\sigma$. We establish
\begin{equation}
 R_c'(\beta)=\frac1\beta
 \mathcal V_{c\beta}^{\rho_\beta^{(c)}}(L_\beta),
 \label{eq:intro-response}
\end{equation}
where $\mathcal V$ is an explicit nonnegative quadratic form
generalizing variance
(Proposition~\ref{thm:escort-susceptibility}). This escort
susceptibility measures the response of the relative-entropy profile
to a change of order. A second integration gives an exact expression
for
$D_{\alpha,c\alpha}(\rho\Vert\sigma)-D(\rho\Vert\sigma)$,
so bounds on the susceptibility translate into quantitative control
of the departure from relative entropy at finite order.

All fixed rays, $z=c\alpha$ with $c$ constant, 
recover relative entropy at order one, but their
first derivatives need not agree. These differences reveal information
that is absent from the common limiting value. Defining
\[
 S_c=\left.\frac{d}{d\alpha}
 D_{\alpha,c\alpha}(\rho\Vert\sigma)\right|_{\alpha=1},
\]
we obtain
\begin{equation}
 S_{1/2}-S_1=\frac18F_Q(\rho,\log\sigma),
 \qquad
 S_1-S_2=\frac14I_{\rm WY}(\rho,\log\sigma).
 \label{eq:intro-slope-identities}
\end{equation}
Here $F_Q$ is the symmetric-logarithmic-derivative quantum Fisher
information for unitary evolution, and $I_{\rm WY}$ is the
Wigner--Yanase skew information. Both quantify sensitivity to the
flow generated by $\log\sigma$, and both vanish precisely when
$\rho$ and $\sigma$ commute. Comparing order derivatives therefore
isolates a specifically noncommutative contribution to
distinguishability. For a Gibbs reference state, this contribution
measures coherence between distinct energy eigenspaces.

The same response formula recovers the known identity
$S_1=\tfrac12\operatorname{Var}_\rho(\log\rho-\log\sigma)$
\cite{Lin:2014properties} and extends
\eqref{eq:intro-slope-identities} to an integer hierarchy involving
Wigner--Yanase--Dyson information. When both states are Gibbs states
at the same temperature, the integer slopes become discrete
imaginary-time averages of an equilibrium correlation function.
Their limit is the Bogoliubov--Kubo--Mori covariance, relating the
order dependence to static thermal response. These differential
results are established for faithful finite-dimensional states.
Extending them to general algebras requires additional regularity
beyond the endpoint finiteness assumed in the integral theorem.

The role of the reference state becomes especially important for
conditional entropies, where it is chosen by optimization. For a faithful bipartite state $\rho_{AB}$, define
\[H_{\alpha,z}^{\uparrow}(A|B)_\rho
=-\inf_{\sigma_B}D_{\alpha,z}(\rho_{AB}\Vert\one_A\otimes\sigma_B).\]
At a minimizing reference state, the corresponding escort has
$B$-marginal equal to that reference state. This converts its relative
entropy into minus its conditional entropy. An envelope argument then
gives an integral representation of $H_{\alpha,c\alpha}^{\uparrow}$
for every $c>0$ and $\alpha>1$, without assuming uniqueness of the
minimizer (Proposition~\ref{prop:optimized-conditional-integral}).
At order one, the optimizing state is $\rho_B$, and differences
of the order-one slopes of the optimized divergence
$-H_{\alpha,c\alpha}^{\uparrow}$ are given by Fisher and skew information for
the observable $\one_A\otimes\log\rho_B$
(Lemma~\ref{lem:optimized-slopes}). 
At $c=1$, this recovers the sandwiched conditional-entropy identity
of Cheng and Tomamichel
\cite[Lemma~7]{Cheng:2026transmission}.  

The description of order dependence clarifies the optimizations
arising in hypothesis testing and work extraction. 
For the established strong-converse exponent
of binary hypothesis testing
\cite{Mosonyi:2013testing,Hiai:2021discrimination,Junge:2025converse},
our integral representation identifies a continuous interior optimizer by
the condition that its escort relative entropy equals the prescribed
testing rate. A corresponding condition holds for the
work-reliability function of
\cite{Watanabe:2026reliability}. In that setting, the susceptibility
also determines the curvature with respect to the work rate and,
at the half-order optimizer, separates out an explicit quantum
Fisher information contribution. These results connect the
sensitivity of an operational exponent to fluctuations in the
escort state selected by the optimization.

The rest of this paper is organized as follows.
Section~\ref{sec:preliminaries} introduces the algebraic definitions,
states the integral representation and its boundary correction,
and presents the elementary matrix calculation, deferring the proof 
of the general algebraic theorem to Section~\ref{sec:proof-alpha-z}. 
We next apply the integral representation to a few settings.
Section~\ref{sec:escort-response} develops the convexity and response
formulas, including their Fisher-information and thermal interpretations.
Section~\ref{sec:optimized-entropies} treats optimized conditional
entropies and classical--quantum ensembles, and
Section~\ref{sec:applications} applies the representation to quantum
field theory, hypothesis testing, and work reliability. We conclude in
Section~\ref{sec:discussion} with discussion of various further questions.
The appendices contain the matrix support arguments and stability
estimates.

\section{Preliminaries and main results}
\label{sec:preliminaries}

We introduce $\alpha$--$z$ divergences and their escort states,
then state the integral representation and its order-monotonicity
consequences on a general von Neumann algebra. The construction uses
Haagerup's $L^p$ spaces, including exponents below one. We recall only
the facts needed in the statements and refer to
\cite{Terp:1981spaces,Kosaki:1984interpolation} for the underlying $L^p$ theory.
The algebraic $\alpha$--$z$ definitions follow
\cite{Kato:2023spaces,Kato:2023divergence,Hiai:2024reversibility}. For the sandwiched
divergence, see also
\cite{Berta:2016weighted,Jencova:2016entropies,Jencova:2017spaces}.

\subsection{\texorpdfstring{$\alpha$--$z$}{alpha-z} and sandwiched R\'enyi divergences}

Let $M$ be a von Neumann algebra and choose a faithful normal semifinite
weight.  Its continuous core $\cN$ carries a canonical semifinite trace
$\tau$ and a dual action $(\theta_s)_{s\in\mathbb R}$.  For $0<p<\infty$,
the Haagerup space $L^p(M)$ consists of the $\tau$-measurable operators
affiliated with $\cN$ that satisfy
\begin{equation}
  \theta_s(x)=\e^{-s/p}x,
  \qquad s\in\mathbb R.
  \label{eq:haagerup-homogeneity}
\end{equation}
The construction is independent of the initial
weight up to canonical isometry.  For $p\geq1$ it is a 
Banach space, while for $0<p<1$ it is a complete
quasi-Banach space. The standard $L^p(M)$ (quasi-)norm is denoted by
$\|\cdot\|_p$. Products of measurable operators are understood as their closed
products. They obey the generalized H\"older inequality.

There is a canonical isometric identification
$M_*\cong L^1(M)$.  In particular, every
$\varphi\in M_*^+$ has a unique density
$h_\varphi\in L^1(M)_+$.
Moreover, there is a canonical functional $\tr$ on $L^1(M)$ such that
\begin{equation}
  \tr(h_\varphi)=\varphi(\one),
  \label{eq:density-normalization}
\end{equation}
and
\begin{equation}
    \tr(|h_{\varphi}|)=\tr(h_{|\varphi|})=|\varphi|(\one)=\|\varphi\|.
\end{equation}
The support of $h_\varphi$ corresponds to the support projection
$s(\varphi)\in M$.  For a projection $e\in M$, there is a canonical
identification
\begin{equation}
  L^p(eMe)\cong eL^p(M)e.
  \label{eq:corner-identification}
\end{equation}
In particular, after replacing $M$ by
$eMe$, $e=s(\omega)$, a normal state $\omega$ becomes faithful.  Since a
faithful normal state always exists on $eMe$, we do not need to impose any 
$\sigma$-finiteness hypothesis on $M$ in our main results.

We write $D(\psi\Vert\omega)$ for Araki relative entropy of normal states,
with value in $[0,+\infty]$.  We extend the first argument to a nonzero
normal positive functional $\nu=q\varphi$, where $q=\nu(\one)>0$ and
$\varphi$ is a state, by the standard scaling law
\begin{equation}
  D(q\varphi\Vert\omega)
  \defeq qD(\varphi\Vert\omega)+q\log q.
  \label{eq:araki-relative-entropy-scaling}
\end{equation}
Its normalized relative entropy is therefore
\begin{equation}
  D_1(\nu\Vert\omega)
  \defeq \frac{D(\nu\Vert\omega)}{\nu(\one)}
  =D(\varphi\Vert\omega)+\log q.
  \label{eq:relative-entropy-scaling}
\end{equation}
The second argument is a state throughout the algebraic discussion.
In Section~\ref{sec:optimized-entropies}, we also use positive matrix
reference operators of the form $\one_A\otimes\sigma_B$, with their
normalization stated explicitly.


Let $\psi$ and $\omega$ be normal states, and put $e=s(\omega)$.  For
$0<\alpha<1$ and $z>0$, define the closed positive product
\begin{equation}
  x_{\alpha,z}(\psi\|\omega)
  \defeq \overline{
    h_\omega^{\frac{1-\alpha}{2z}}
    h_\psi^{\frac\alpha z}
    h_\omega^{\frac{1-\alpha}{2z}}
  }\in eL^z(M)_+e
  \label{eq:lower-alpha-z-factor}
\end{equation}
and
\begin{equation}
  Q_{\alpha,z}(\psi\|\omega)
  \defeq \tr\bigl(x_{\alpha,z}(\psi\|\omega)^z\bigr).
  \label{eq:lower-alpha-z-moment}
\end{equation}
For $\alpha>1$, the moment $Q_{\alpha,z}(\psi\|\omega)$ is finite when there is a positive factor \cite{Kato:2023spaces,Kato:2023divergence}
\begin{equation}
  x_{\alpha,z}(\psi\|\omega)\in eL^z(M)_+e
  \label{eq:upper-alpha-z-corner}
\end{equation}
satisfying
\begin{equation}
  h_\psi^{\alpha/z}
  =h_\omega^{\frac{\alpha-1}{2z}}
   x_{\alpha,z}(\psi\|\omega)
   h_\omega^{\frac{\alpha-1}{2z}}.
  \label{eq:upper-alpha-z-factor}
\end{equation}
The factor is unique in the support corner.  In this case
\begin{equation}
  Q_{\alpha,z}(\psi\|\omega)
  \defeq \tr\bigl(x_{\alpha,z}(\psi\|\omega)^z\bigr);
  \label{eq:upper-alpha-z-moment}
\end{equation}
otherwise it is $+\infty$.  The {\az} divergence is
\begin{equation}
  D_{\alpha,z}(\psi\Vert\omega)
  \defeq \frac{1}{\alpha-1}
  \log Q_{\alpha,z}(\psi\|\omega).
  \label{eq:algebraic-alpha-z-divergence}
\end{equation}

The same factor and moment definitions apply to a nonzero normal positive
first argument $\nu$ after replacing $h_\psi$ by $h_\nu$.  In particular,
for $q>0$ and a state $\varphi$,
\begin{equation}
  Q_{\alpha,z}(q\varphi\|\omega)
  =q^\alpha Q_{\alpha,z}(\varphi\|\omega).
  \label{eq:alpha-z-moment-homogeneity}
\end{equation}
For such an argument we employ the normalized convention
\begin{equation}
  D_{\alpha,z}(\nu\Vert\omega)
  \defeq \frac{1}{\alpha-1}
  \log\frac{Q_{\alpha,z}(\nu\|\omega)}{\nu(\one)}.
  \label{eq:normalized-alpha-z-positive}
\end{equation}
For every fixed $z>0$, its lower-order limit is
\begin{equation}
  \lim_{\alpha\uparrow1}D_{\alpha,z}(\nu\Vert\omega)
  =D_1(\nu\Vert\omega).
  \label{eq:lower-alpha-z-limit}
\end{equation}
The upper-order limit $\alpha\downarrow1$ is known under the additional
conditions $z>1/2$ and
$D_{\alpha_0,z}(\nu\Vert\omega)<\infty$ for some
$\alpha_0\in(1,2z]$ \cite[Theorems~6.7 and~6.9]{Hiai:2024reversibility}.

The sandwiched quantities are the specialization
\begin{equation}
  Q_\alpha(\psi\|\omega)
  \defeq Q_{\alpha,\alpha}(\psi\|\omega),
  \qquad
  D_\alpha(\psi\Vert\omega)
  \defeq D_{\alpha,\alpha}(\psi\Vert\omega).
  \label{eq:srd-specialization}
\end{equation}
For $0<\alpha<1$ this means
\begin{equation}
  x_\alpha^s(\psi\|\omega)
  \defeq \overline{
  h_\omega^{\frac{1-\alpha}{2\alpha}}
  h_\psi
  h_\omega^{\frac{1-\alpha}{2\alpha}}},
  \qquad
  Q_\alpha(\psi\|\omega)
  =\tr\bigl(x_\alpha^s(\psi\|\omega)^\alpha\bigr).
  \label{eq:lower-srd-factor}
\end{equation}
For $\alpha>1$ with finite moment, let
$x_\alpha^s(\psi\|\omega)\in eL^\alpha(M)_+e$ be the upper factor
defined by
\begin{equation}
  h_\psi
  =h_\omega^{\frac{\alpha-1}{2\alpha}}
   x_\alpha^s(\psi\|\omega)
   h_\omega^{\frac{\alpha-1}{2\alpha}},
  \qquad
  Q_\alpha(\psi\|\omega)
  =\tr\bigl(x_\alpha^s(\psi\|\omega)^\alpha\bigr).
  \label{eq:upper-srd-factor}
\end{equation}
At the endpoint we set
\begin{equation}
  x_1^s(\psi\|\omega)\defeq h_\psi.
  \label{eq:srd-factor-at-one}
\end{equation}

The full data-processing region for $D_{\alpha,z}$ is \cite{Zhang:2018conjecture,Hiai:2024reversibility}
\begin{equation}
  \begin{cases}
  0<\alpha<1,
  &z\geq\max\{\alpha,1-\alpha\},\\[0.4ex]
  \alpha>1,
  &\max\{\alpha/2,\alpha-1\}\leq z\leq\alpha.
  \end{cases}
  \label{eq:dpi-region}
\end{equation}
Note that the integral theorem holds for every $z>0$. Only the applications that invoke data processing require this
parameter restriction. 

\subsection{Fixed-ray factors and escort trajectories}
\label{subsec:algebraic-ray-escorts}

Fix $c>0$ and consider the ray $(\beta,z)=(\beta,c\beta)$.
Along this ray the exponent of $h_\psi$ is fixed at $1/c$.
Write $e=s(\omega)$.

For $0<\beta<1$, define
\begin{equation}
  x_\beta^{(c)}(\psi\|\omega)
  \defeq \overline{
  h_\omega^{\frac{1-\beta}{2c\beta}}
  h_\psi^{1/c}
  h_\omega^{\frac{1-\beta}{2c\beta}}}
  \in eL^{c\beta}(M)_+e.
  \label{eq:lower-fixed-ray-factor}
\end{equation}

At $\beta=1$, if $s(\psi)\leq e$, we set
$x_1^{(c)}=h_\psi^{1/c}\in eL^c(M)e$ and
$Q_1^{(c)}(\psi\|\omega)=Q_{1,c}(\psi\|\omega)\defeq 1$. If
support inclusion does not hold, the order-one value must be 
distinguished from the boundary value
of the lower-order construction, which remains confined to the
corner $eMe$. 

For $\beta>1$, whenever it exists, let
$x_\beta^{(c)}\in eL^{c\beta}(M)_+e$ be the unique positive factor such that
\begin{equation}
  h_\psi^{1/c}
  =h_\omega^{\frac{\beta-1}{2c\beta}}
   x_\beta^{(c)}
   h_\omega^{\frac{\beta-1}{2c\beta}}.
  \label{eq:upper-fixed-ray-factor}
\end{equation}
In all cases put
\begin{equation}
  Q_\beta^{(c)}(\psi\|\omega)
  \defeq \tr\bigl((x_\beta^{(c)})^{c\beta}\bigr)
  =Q_{\beta,c\beta}(\psi\|\omega),
  \label{eq:algebraic-ray-moment}
\end{equation}
with $ Q_\beta^{(c)}(\psi\|\omega)=+\infty$ when the factor $x_\beta^{(c)}$ does not exist.
Whenever this number is finite and nonzero, the fixed-ray escort state
$\psi_\beta^{(c)}$ is defined in terms of its associated $L^1$ density
\begin{equation}
  h_{\psi_\beta^{(c)}}
  \defeq \frac{(x_\beta^{(c)})^{c\beta}}
          {Q_\beta^{(c)}(\psi\|\omega)} \ .
  \label{eq:algebraic-ray-escort}
\end{equation}
The value $c=1$ gives the sandwiched factors and escorts.  We set
$\psi_\beta^s\defeq \psi_\beta^{(1)}$ and have
\begin{equation}
  x_\beta^{(1)}=x_\beta^s,
  \qquad
  Q_\beta^{(1)}=Q_\beta,
  \qquad \psi_\beta^{(1)}=\psi_\beta^s.
  \label{eq:sandwiched-ray-specialization}
\end{equation}

For faithful density matrices $\rho,\sigma$, these definitions take
the explicit form
\begin{align}
  X_{\beta}^{(c)}
  &\defeq
  \sigma^{\frac{1-\beta}{2c\beta}}
  \rho^{1/c}
  \sigma^{\frac{1-\beta}{2c\beta}},
  \label{eq:fd-fixed-ray-factor}\\
  Q_{\beta}^{(c)}(\rho\|\sigma)
  &\defeq
  \Tr\!\left[(X_{\beta}^{(c)})^{c\beta}\right]
  =Q_{\beta,c\beta}(\rho\|\sigma),
  \label{eq:fd-fixed-ray-moment}\\
  \rho_{\beta}^{(c)}
  &\defeq
  \frac{(X_{\beta}^{(c)})^{c\beta}}
       {Q_{\beta}^{(c)}(\rho\|\sigma)}.
  \label{eq:fd-fixed-ray-escort}
\end{align}
Here all powers are ordinary matrix powers. In particular,
$Q_1^{(c)}=1$, $\rho_1^{(c)}=\rho$, and we write
$\rho_\beta^s=\rho_\beta^{(1)}$ on the sandwiched ray.

\subsection{Integral representation theorem}

We now state our main structural results.
We give a simple proof of the integral representation 
of Theorem~\ref{thm:fixed-ray-integral} for the finite dimensional 
case in Section \ref{sec:finite-dimensional}.
Since the general proof is technical, we defer it to Section
\ref{sec:proof-alpha-z}.
Theorem~\ref{thm:fixed-ray-integral} and
Corollary~\ref{cor:raywise-monotonicity} are proved together in
Section~\ref{subsec:fixed-ray-proof}, using the weighted entropy identity
of Section~\ref{subsec:compatible-scale}. Proposition 
\ref{prop:general-fixed-ray-boundary}  is proved
in Section~\ref{subsec:lower-order-proof}.

For $\alpha>0$, $\alpha\ne1$, define the interval
\begin{equation}
  I_\alpha
  \defeq \begin{cases}
     [1,\alpha],&\alpha>1,\\
     [\alpha,1],&0<\alpha<1,
     \end{cases}
  \label{eq:order-interval}
\end{equation}
and the probability measure
\begin{equation}
  \dd\mu_\alpha(\beta)
  \defeq \begin{cases}
  \displaystyle
  \frac{\alpha}{\alpha-1}\frac{\dd\beta}{\beta^2},
  &\alpha>1,\\[1.2ex]
  \displaystyle
  \frac{\alpha}{1-\alpha}\frac{\dd\beta}{\beta^2},
  &0<\alpha<1.
  \end{cases}
  \label{eq:order-probability-measure}
\end{equation}

\begin{theorem}[Fixed-ray $\alpha$--$z$ integral representation]
\label{thm:fixed-ray-integral}
Let $M$ be an arbitrary von Neumann algebra and let $\psi,\omega$ be
normal states.  Fix $\alpha,z>0$, $\alpha\ne1$, and put $c=z/\alpha$.
Assume, branchwise, that
\begin{equation}
  \begin{cases}
  s(\psi)\leq s(\omega),&0<\alpha<1,\\
  Q_{\alpha,z}(\psi\|\omega)<\infty,&\alpha>1.
  \end{cases}
  \label{eq:main-branchwise-hypotheses}
\end{equation}
Then,
for every $\beta\in I_\alpha$, the fixed-ray factor $x_\beta^{(c)}$ exists
and its moment is finite and nonzero, so that
\eqref{eq:algebraic-ray-escort} defines a normal state.  Moreover,
\begin{equation}
  D_{\alpha,z}(\psi\Vert\omega)
  =\frac{\alpha}{\alpha-1}
   \int_1^\alpha
   \frac{D(\psi_\beta^{(c)}\Vert\omega)}{\beta^2}
   \,\dd\beta
  =\int_{I_\alpha}
   D(\psi_\beta^{(c)}\Vert\omega)\,
   \dd\mu_\alpha(\beta).
  \label{eq:main-fixed-ray-integral}
\end{equation}
\end{theorem}
Note that in the upper branch, finiteness already implies support inclusion.  
There is also a sharp lower-order formula without support inclusion.  
\begin{proposition}[General lower-order fixed-ray boundary]
\label{prop:general-fixed-ray-boundary}
Let $M$ be a von Neumann algebra, let $\psi,\omega$ be normal states, let
$0<\alpha<1$, and let $c>0$.  Put
$e=s(\omega)$, $\omega_e\defeq \omega\vert_{eMe}$, and, under the canonical
identification $L^c(eMe)\cong eL^c(M)e$, define
\begin{equation}
  y_c\defeq \overline{e h_\psi^{1/c}e}\in L^c(eMe)_+,
  \qquad
  q_c\defeq \tr_{eMe}(y_c^c).
  \label{eq:general-boundary-anchor}
\end{equation}
If $q_c=0$, then
$D_{\alpha,c\alpha}(\psi\Vert\omega)=+\infty$.  If $q_c>0$, define a
normal state $\widehat\psi_c$ on $eMe$ by
\begin{equation}
  h_{\widehat\psi_c}\defeq q_c^{-1}y_c^c.
  \label{eq:general-boundary-seed}
\end{equation}
Then, for every $0<\beta<1$, the fixed-ray escort of $(\psi,\omega)$ is
the canonical extension to $M$ of the fixed-ray escort of
$(\widehat\psi_c,\omega_e)$, and
\begin{equation}
  D_{\alpha,c\alpha}(\psi\Vert\omega)
  =\frac{\alpha}{1-\alpha}
   \int_\alpha^1
   \frac{D((\widehat\psi_c)_\beta^{(c)}\Vert\omega_e)}{\beta^2}
   \,\dd\beta
   -\frac{\alpha}{1-\alpha}\log q_c.
  \label{eq:general-fixed-ray-boundary}
\end{equation}
Moreover,
\begin{equation}\label{eq:boundary-cost-bound}
    0\leq q_c \leq 1.
\end{equation}
\end{proposition}

The boundary term is already necessary for a two-point classical
probability distribution.
Let $\rho=(r,1-r)$ and $\sigma=(1,0)$ on $\mathbb C^2$, where
$0<r<1$.  On the ray $c=1$, the $\alpha$--$z$ divergence reduces to the sandwiched R\'enyi divergence $D_{\alpha}$, and every interior escort with
$0<\beta<1$ equals $\sigma$.  The escort integral therefore vanishes, but
\begin{equation}
  D_\alpha(\rho\Vert\sigma)
  =-\frac{\alpha}{1-\alpha}\log r,
  \qquad 0<\alpha<1.
\end{equation}
This is precisely the support-boundary term with
$r=\rho(s(\sigma))$.


For a support-nested pair of states $\psi,\omega$ with $s(\psi)\leq s(\omega)$, and fixed $c>0$, the map
$\beta\mapsto\psi_\beta^{(c)}$ is the \emph{fixed-ray escort trajectory}.
Use the endpoint convention $Q_{1,c}=Q_1^{(c)}=1$ from
Section~\ref{subsec:algebraic-ray-escorts}, and define
\begin{equation}
  \begin{gathered}
  \mathcal J_c
  \defeq\{\beta>0:0<Q_{\beta,c\beta}(\psi\|\omega)<\infty\},\\
  R_c(\beta)\defeq D(\psi_\beta^{(c)}\Vert\omega),
  \qquad
  F_c(\beta)\defeq\frac1\beta
       \log Q_{\beta,c\beta}(\psi\|\omega),
  \qquad \beta\in\mathcal J_c.
  \end{gathered}
  \label{eq:escort-profile}
\end{equation}
For $c=1$, $R_1(\beta)$ is the object referred to as refined SRD in \cite{Bao:2019holographic}. We will sometimes refer to $R_c(\beta)$ as \textit{fixed-ray refined divergence}.  For $\beta\neq1$,
\[
  F_c(\beta)=\frac{\beta-1}{\beta}
  D_{\beta,c\beta}(\psi\Vert\omega),
  \qquad F_c(1)=0.
\]

Supportedness $s(\psi)\leq s(\omega)$ gives $(0,1]\subseteq\mathcal J_c$ using Theorem~\ref{thm:fixed-ray-integral}.  If
$\beta_2>1$ belongs to $\mathcal J_c$, Theorem~\ref{thm:fixed-ray-integral}
applied with endpoint $\beta_2$ shows that
$[1,\beta_2]\subseteq\mathcal J_c$.  Hence $\mathcal J_c$ is an interval.
The proof of Theorem~\ref{thm:fixed-ray-integral} in
Section~\ref{subsec:fixed-ray-proof} also yields the following raywise
monotonicity and derivative identity. Let $\frac{\dd^-}{\dd\beta}$ denote the left derivative with respect to $\beta$. 

\begin{corollary}[Raywise refinement and monotonicity]
\label{cor:raywise-monotonicity}
The function $R_c$ is nondecreasing on $\mathcal J_c$.  At every interior
point of $\mathcal J_c$,
\begin{equation}
  \beta^2\frac{\dd^-}{\dd\beta}F_c(\beta)
  =R_c(\beta).
  \label{eq:raywise-refined-identity}
\end{equation}
At $\beta=1$ the same identity holds with the left derivative at the
endpoint, even when $\mathcal J_c=(0,1]$; if
$D(\psi\Vert\omega)=+\infty$, this derivative is understood in the
extended-real sense.  
In finite dimensions these functions are smooth on the interior of
$\mathcal J_c$, and the left derivative there may be replaced by the
ordinary derivative.
\end{corollary}

Thus the finite-dimensional refined divergence of \cite{Bao:2019holographic} extends intrinsically as a
one-sided slope along the escort trajectory.  Since $\mu_\alpha$ is a probability measure, for every $\alpha\in \mathcal J_c\backslash\{1\} $, monotonicity
also gives the endpoint bounds
\begin{align}
  D(\psi\Vert\omega)
  &\leq D_{\alpha,c\alpha}(\psi\Vert\omega)
  \leq D(\psi_\alpha^{(c)}\Vert\omega),
  &&\alpha>1,
  \label{eq:upper-endpoint-bounds}\\
  D(\psi_\alpha^{(c)}\Vert\omega)
  &\leq D_{\alpha,c\alpha}(\psi\Vert\omega)
  \leq D(\psi\Vert\omega),
  &&0<\alpha<1,
  \label{eq:lower-endpoint-bounds}
\end{align}
in the extended sense.

The derivative identity and monotonicity are established together
with Theorem~\ref{thm:fixed-ray-integral} in
Section~\ref{subsec:fixed-ray-proof}. Their convexity interpretation
is developed in Section~\ref{subsec:order-convexity}.

\subsection*{A matrix calculation}
\label{sec:finite-dimensional}

For faithful density matrices $\rho,\sigma$ and a fixed $c>0$,
the representation follows by differentiating under the trace.

Put $m_c(\beta)=Q_{\beta}^{(c)}(\rho\|\sigma)$ and
$g_c(\beta)=\log m_c(\beta)$.  Functional calculus gives
\begin{equation}
  \frac{\dd X_{\beta}^{(c)}}{\dd\beta}
  =-\frac{1}{2c\beta^2}
  \left[(\log\sigma)X_{\beta}^{(c)}
       +X_{\beta}^{(c)}(\log\sigma)\right].
  \label{eq:fd-factor-derivative}
\end{equation}
Using tracial differentiation and cyclicity,
\begin{equation}
  m_c'(\beta)
  =c\Tr\!\left[(X_{\beta}^{(c)})^{c\beta}
                  \log X_{\beta}^{(c)}\right]
   -\frac1\beta\Tr\!\left[(X_{\beta}^{(c)})^{c\beta}
                  \log\sigma\right].
  \label{eq:fd-moment-derivative}
\end{equation}
On the other hand,
\begin{equation}
  \log\rho_{\beta}^{(c)}
  =c\beta\log X_{\beta}^{(c)}-g_c(\beta),
\end{equation}
and therefore
\begin{equation}
  D(\rho_{\beta}^{(c)}\Vert\sigma)
  =\beta g_c'(\beta)-g_c(\beta).
  \label{eq:fd-refined-identity}
\end{equation}
Equivalently,
\begin{equation}
  \beta^2\frac{\dd}{\dd\beta}
  \left[
    \frac{\beta-1}{\beta}
    D_{\beta,c\beta}(\rho\Vert\sigma)
  \right]
  =D(\rho_{\beta}^{(c)}\Vert\sigma).
  \label{eq:fd-ray-refined-identity}
\end{equation}
The quantity in square brackets vanishes at $\beta=1$. Integrating
therefore gives the matrix form of Theorem~\ref{thm:fixed-ray-integral}.

For $c=1$, \eqref{eq:fd-ray-refined-identity} is the refined
sandwiched identity of \cite{Bao:2019holographic}.
For singular states with $s(\rho)\leq s(\sigma)$, first pass to
the support of $\sigma$, then approximate $\rho$ by faithful states.
Appendix~\ref{app:matrix-support} gives this argument and the matrix
boundary formula when the supports are not nested.

The preceding calculation uses matrix logarithms and differentiation
of traces of powers. On a general von Neumann algebra the factors can
be unbounded, and their Haagerup exponents change with the order.
The weighted scale constructed in Section \ref{sec:proof-alpha-z} 
supplies the entropy identity
and endpoint continuity needed to replace this calculation.

\section{Order dependence and escort response}
\label{sec:escort-response}

The representation gives a convex potential whose slope is the
relative entropy of the escort state. We first describe this
structure on a general von Neumann algebra, then compute its
curvature for faithful density matrices. Differences of the
order-one response coefficients recover quantum Fisher information
and Wigner--Yanase skew information.

\subsection{A convex potential for the R\'enyi order}
\label{subsec:order-convexity}

The integral representation admits a useful geometric interpretation.
After a change of variable, the escort relative entropy becomes the
slope of a convex function. This formulation will also be useful for
optimizing over the R\'enyi order.

Let $\psi,\omega$ be normal states on a von Neumann algebra with
$s(\psi)\leq s(\omega)$, and fix $c>0$. Recall the escort
relative-entropy profile $R_c$ and the finite-moment domain
$\mathcal J_c$ from \eqref{eq:escort-profile}. Introduce
\begin{equation}
 \kappa\defeq 1-\frac1\alpha,
 \qquad
 \alpha=\frac1{1-\kappa},
 \qquad -\infty<\kappa<1.
 \label{eq:kappa-parameter}
\end{equation}
This parameter is chosen to absorb the weight in the integral
representation: under $u=1-1/\beta$, we have
$\dd u=\beta^{-2}\dd\beta$. Moreover,
$\alpha/(\alpha-1)=1/\kappa$.

The orders for which the moment is finite correspond to the interval
\begin{equation}
 \mathcal K_c
 \defeq
 \left\{\kappa<1:
       \frac1{1-\kappa}\in\mathcal J_c\right\}.
 \label{eq:ray-effective-domain}
\end{equation}
The support assumption ensures that
$(-\infty,0]\subseteq\mathcal K_c$.
On this interval, define
\begin{equation}
 \Phi_c(\kappa)
 \defeq
 \kappa D_{\frac1{1-\kappa},\frac{c}{1-\kappa}}
             (\psi\Vert\omega)
 \quad (\kappa\ne0),
 \qquad
 \Phi_c(0)\defeq0.
 \label{eq:ray-potential}
\end{equation}
Multiplying \eqref{eq:main-fixed-ray-integral} by $\kappa$ and
making the above change of variable gives
\begin{equation}
 \Phi_c(\kappa)
 =
 \int_0^\kappa
 R_c\!\left(\frac1{1-u}\right)\,\dd u,
 \qquad \kappa\in\mathcal K_c.
 \label{eq:ray-potential-integral}
\end{equation}

Both $R_c$ (Corollary~\ref{cor:raywise-monotonicity}) and $u\mapsto(1-u)^{-1}$ are nondecreasing.
The integrand in \eqref{eq:ray-potential-integral} is therefore
nondecreasing, which proves that $\Phi_c$ is convex.
More precisely, the anchored entropy identity
(Proposition~\ref{prop:anchored-entropy-identity}) gives
\begin{equation}
    \Phi'_{c,-}(\kappa)
    =
    R_c\!\left(\frac1{1-\kappa}\right)
\end{equation}
at interior points of $\mathcal K_c$.
At $\kappa=0$, the left derivative equals
$D(\psi\Vert\omega)$, with the endpoint interpretation given in
Corollary~\ref{cor:raywise-monotonicity}.
Thus the escort relative entropy is the left tangent slope of
$\Phi_c$, whereas
\begin{equation}
    D_{\alpha,c\alpha}(\psi\Vert\omega)
    =
    \frac{\Phi_c(\kappa)-\Phi_c(0)}{\kappa},
    \qquad
    \kappa=1-\frac1\alpha\ne0,
\end{equation}
is the slope between the origin and the point
$(\kappa,\Phi_c(\kappa))$.
Convexity expresses the fact that the escort relative entropy
increases along the ray.

For faithful density matrices, the response formulas below identify
the curvature of $\Phi_c$ explicitly.
Section~\ref{subsec:binary-testing} uses this potential to study
optimization over the R\'enyi order.
When needed, we regard $\Phi_c$ as an extended-real-valued convex
function on $(-\infty,1)$ by setting it equal to $+\infty$
outside $\mathcal K_c$.

\subsection{The response of an escort trajectory}

We have shown in Corollary \ref{cor:raywise-monotonicity} that the relative entropy of the escort state is
nondecreasing along each fixed ray. We now determine its rate of
increase and relate this response to fluctuations in the escort
state.

Throughout this section, $\rho$ and $\sigma$ are faithful density
matrices on a finite-dimensional Hilbert space. All matrix powers
and derivatives used below are therefore well defined for every
positive order. We write
\(
 R_c(\beta)=D(\rho_\beta^{(c)}\Vert\sigma)
\)
as in \eqref{eq:escort-profile}
for the escort relative-entropy profile along the ray $z=c\beta$.

The commuting case illustrates the relation to fluctuations.
If $\rho$ and $\sigma$ have eigenvalues $r_i$ and $s_i$ in a common
eigenbasis, their escort state has probabilities
\[
 q_i(\beta)
 =
 \frac{r_i^\beta s_i^{1-\beta}}
      {\sum_j r_j^\beta s_j^{1-\beta}},
\]
independently of $c$. Direct differentiation gives
\[
 R_c'(\beta)
 =
 \beta\,\operatorname{Var}_{q(\beta)}
       \!\left(\log\frac{r_i}{s_i}\right)
 =
 \frac1\beta\,\operatorname{Var}_{q(\beta)}
       \!\left(\log\frac{q_i(\beta)}{s_i}\right),
\]
where
$\operatorname{Var}_{q}(a)
=\sum_i q_i a_i^2-(\sum_i q_i a_i)^2$.
Thus the rate of increase is determined by fluctuations of the
log-likelihood ratio, evaluated in the current escort state.

For noncommuting states, we obtain a corresponding identity in
which the ordinary variance is replaced by a positive quadratic
form. This form determines both the response of $R_c$ and the
curvature of the potential $\Phi_c$ introduced in
Section~\ref{subsec:order-convexity}. At order one, comparing the
response along different rays yields quantum Fisher information
and Wigner--Yanase skew information.


Differentiating the matrix powers in the escort state leads to a
spectral multiplier whose kernel we now define. For $t\geq1$, this
kernel coincides with the symmetric homogeneous mean $A_{1/t}$
studied by Furuichi \cite{Furuichi:2014means}. Its centered quadratic
form has the structure of a generalized quantum covariance
\cite{Hiai:2012entropy}. We use the same formula for all $t>0$
and identify its role in the response of the escort trajectory.
For a positive definite matrix
$\tau=\sum_i p_i|i\rangle\langle i|$ and $t>0$, set\footnote{Note the useful identity $\displaystyle \frac{x^{1/n}+y^{1/n}}{x^{1/n}-y^{1/n}}=\coth\left( \frac{\log x-\log y}{2n} \right)$ for $x\neq y$ and positive integer $n$.}
\begin{align}
 m_t(x,y)
 &\defeq\frac{x-y}{2t}
       \coth\!\left(\frac{\log x-\log y}{2t}\right),
       \qquad m_t(x,x)\defeq x,
       \label{eq:response-kernel}\\
 [\mathcal M_t^\tau(A)]_{ij}
 &\defeq m_t(p_i,p_j)A_{ij}.
       \label{eq:response-multiplier}
\end{align}
The continuous value at $x=y$ makes the multiplier independent of
the choice of basis within a degenerate eigenspace. For a density
matrix $\tau$ and a Hermitian matrix $A$, define
\begin{equation}
 \mathcal V_t^\tau(A)
 \defeq \Tr\bigl[A\mathcal M_t^\tau(A)\bigr]
             -(\Tr\tau A)^2
 =\sum_{i,j}m_t(p_i,p_j)|A_{ij}|^2-(\Tr\tau A)^2.
 \label{eq:response-variance}
\end{equation}
This form is nonnegative. To see this, note first that 
$m_t(p_i,p_i)=p_i$ implies
$\mathcal M_t^\tau(\one)=\tau$ and
$\Tr\mathcal M_t^\tau(A)=\Tr\tau A$. Therefore 
\[
\Tr([A-(\Tr\tau A)\one]\mathcal M_t^\tau(A-(\Tr\tau A)\one))=\Tr\bigl[A\mathcal M_t^\tau(A)\bigr]
             -(\Tr\tau A)^2,
\]
which gives
\[
 \mathcal V_t^\tau(A)
 =\sum_{i,j}m_t(p_i,p_j)
   |[A-(\Tr\tau A)\one]_{ij}|^2\geq0.
\]
Since every $m_t(p_i,p_j)$ is strictly positive, equality holds
precisely when $A$ is scalar. At $t=1$, $m_1(x,y)=(x+y)/2$, and 
$\mathcal{V}_1^\tau(A)=\Tr \tau A^2-(\Tr\tau A)^2$ becomes the 
ordinary variance.

For the escort $\rho_\beta^{(c)}$ in
\eqref{eq:fd-fixed-ray-escort}, put
\begin{equation}
 L_\beta^{(c)}\defeq\log\rho_\beta^{(c)}-\log\sigma.
 \label{eq:response-profile-notation}
\end{equation}

\begin{proposition}[Escort susceptibility]
\label{prop:escort-susceptibility}
\label{thm:escort-susceptibility}
For every $c,\beta>0$, the escort trajectory satisfies
\begin{align}
 \beta\frac{\dd}{\dd\beta}\rho_\beta^{(c)}
 &=\mathcal M_{c\beta}^{\rho_\beta^{(c)}}
       \bigl(L_\beta^{(c)}-R_c(\beta)\one\bigr),
       \label{eq:escort-state-response}\\
 R_c'(\beta)
 &=\frac1\beta\,
       \mathcal V_{c\beta}^{\rho_\beta^{(c)}}
                   (L_\beta^{(c)}).
       \label{eq:escort-profile-response}
\end{align}
If $\rho\ne\sigma$, then $R_c'(\beta)>0$ for every $\beta>0$.
\end{proposition}

\begin{proof}
Fix $c>0$ and set
\[
 K=\sigma^{-1/(2c)}\rho^{1/c}\sigma^{-1/(2c)},\qquad
 Y_p=\sigma^{1/(2p)}K\sigma^{1/(2p)},\qquad
 A_p=Y_p^p,\qquad Z_p=\Tr A_p.
\]
Then $Y_{c\beta}=X_\beta^{(c)}$ and
$\rho_\beta^{(c)}=A_{c\beta}/Z_{c\beta}$.
Writing $B=\log\sigma$, differentiation gives
\[
 Y_p'=-\frac{BY_p+Y_pB}{2p^2}.
\]
For the derivative of $A_p$, we have
\[
    A_p' = Y_p^p\log Y_p + \partial f_p(Y_p)[Y_p'], \qquad f_p(x)=x^p,
\]
with $\partial(f_p)$ the Fr\'echet derivative.
At a fixed value of $p$, choose an
eigenbasis of $Y_p$, with eigenvalues $y_i>0$. The Fr\'echet
derivative of the power map in the matrix variable is
\[
 [\partial f_p(Y^p)[H]]_{ij}
 =\frac{y_i^p-y_j^p}{y_i-y_j}H_{ij},
\]
where the quotient has its continuous value when $y_i=y_j$.
Writing $a_i=y_i^p$, the contribution of $Y_p'$ to $pA_p'$ is
\[
 -\frac{a_i-a_j}{2p}\frac{y_i+y_j}{y_i-y_j}B_{ij}
 =-m_p(a_i,a_j)B_{ij}.
\]
Combining the two contributions gives
\begin{equation}
 pA_p'=A_p\log A_p-\mathcal M_p^{A_p}(B).
 \label{eq:unnormalized-response-identity}
\end{equation}

The kernel is homogeneous:
$m_p(\lambda x,\lambda y)=\lambda m_p(x,y)$ for $\lambda>0$.
Put $\tau_p=A_p/Z_p$. Taking the trace of
\eqref{eq:unnormalized-response-identity} gives
\[
 p\frac{Z_p'}{Z_p}
 =\log Z_p+D(\tau_p\Vert\sigma).
\]
We next differentiate $\tau_p=A_p/Z_p$. Homogeneity, together with
$\mathcal M_p^{\tau_p}(\log\tau_p)=\tau_p\log\tau_p$, gives
\[
 p\tau_p'
 =\mathcal M_p^{\tau_p}
   \bigl(\log\tau_p-\log\sigma-D(\tau_p\Vert\sigma)\one\bigr).
\]
The substitution $p=c\beta$ proves
\eqref{eq:escort-state-response}.

Since $\Tr\frac{\dd}{\dd\beta}\rho_\beta^{(c)}=0$,
tracial differentiation of relative entropy gives
\[
 R_c'(\beta)
 =\Tr\!\left[
   \frac{\dd\rho_\beta^{(c)}}{\dd\beta}L_\beta^{(c)}
   \right].
\]
Equation~\eqref{eq:escort-profile-response} follows from
\eqref{eq:escort-state-response} and
\eqref{eq:response-variance}.
For the strict inequality, the quadratic form vanishes only when
$L_\beta^{(c)}$ is scalar. Normalization then implies
$\rho_\beta^{(c)}=\sigma$. Taking the $(c\beta)$-th root in the escort
definition and cancelling the outer powers of $\sigma$ shows that
$\rho$ is proportional to $\sigma$. Normalization again gives
$\rho=\sigma$.
\end{proof}

A further integration expresses the difference between the R\'enyi
divergence and relative entropy in terms of this quadratic form.
Equation~\eqref{eq:finite-order-response-remainder} also transfers
bounds on the susceptibility to bounds on the divergence at a fixed
order. 

\begin{corollary}[Finite-order remainder and convexity]
\label{cor:response-remainder}
For every $\alpha>0$, $\alpha\ne1$,
\begin{equation}
 D_{\alpha,c\alpha}(\rho\Vert\sigma)-D(\rho\Vert\sigma)
 =\frac1{\alpha-1}\int_1^\alpha
   \frac{\alpha-\beta}{\beta^2}
   \mathcal V_{c\beta}^{\rho_\beta^{(c)}}
                      (L_\beta^{(c)})\dd\beta.
 \label{eq:finite-order-response-remainder}
\end{equation}
The integral is oriented when $\alpha<1$.
For the order potential $\Phi_c$ in \eqref{eq:ray-potential}, with
$\beta=(1-\kappa)^{-1}$ and $\kappa<1$, one has
\begin{equation}
 \Phi_c'(\kappa)=R_c(\beta),\qquad
 \Phi_c''(\kappa)
 =\beta\,\mathcal V_{c\beta}^{\rho_\beta^{(c)}}
                              (L_\beta^{(c)}).
 \label{eq:convex-potential-susceptibility}
\end{equation}
Thus the order potential is strictly convex whenever $\rho\ne\sigma$.
\end{corollary}

\begin{proof}
Substitute
$R_c(\beta)=D(\rho\Vert\sigma)+\int_1^\beta R_c'(u)\dd u$
into the escort representation and integrate by parts to obtain
\eqref{eq:finite-order-response-remainder}.
The substitution $\beta=(1-\kappa)^{-1}$ in the refined identity
gives $\Phi_c'(\kappa)=R_c(\beta)$. Differentiating once more and
using Proposition~\ref{prop:escort-susceptibility} proves
\eqref{eq:convex-potential-susceptibility} and strict convexity.
\end{proof}


\subsection{Order derivatives and quantum Fisher information}

We now evaluate the response at order one and relate it to skew and Fisher information. Relations between R\'enyi divergences and skew or Fisher information
have also been obtained by perturbing the states
\cite{Pires:2020coherences,Wilde:2025fisher}. Here both states are fixed.  

Define
\begin{equation}
 S_c(\rho\Vert\sigma)
 \defeq\left.
   \frac{\dd}{\dd\alpha}D_{\alpha,c\alpha}(\rho\Vert\sigma)
   \right|_{\alpha=1}.
 \label{eq:order-slope-definition}
\end{equation}
Although we have defined $S_c$ along a ray, the first derivative
depends only on the value of the second parameter at $\alpha=1$.
Indeed, the removable extension of $D_{\alpha,z}$ satisfies
$D_{1,z}=D(\rho\Vert\sigma)$ for every $z>0$. Thus
$\partial_zD_{1,z}=0$, and the chain rule gives, for every
differentiable path $z(\alpha)$ with $z(1)=c$,
\[
 \left.\frac{\dd}{\dd\alpha}
 D_{\alpha,z(\alpha)}(\rho\Vert\sigma)\right|_{\alpha=1}
 =\left.\partial_\alpha D_{\alpha,z}(\rho\Vert\sigma)
   \right|_{(\alpha,z)=(1,c)}=S_c(\rho\Vert\sigma).
\]
In particular, changing the tangent of the path leaves $S_c$
unchanged. The same chain-rule observation is used at $c=1$ in
\cite[Proposition~11]{Lin:2014properties}.

Since $\rho_1^{(c)}=\rho$, expansion of
\eqref{eq:finite-order-response-remainder} at $\alpha=1$ gives
\begin{equation}
 S_c(\rho\Vert\sigma)
 =\frac12R_c'(1)
 =\frac12\mathcal V_c^\rho(\log\rho-\log\sigma).
 \label{eq:order-slope-variance}
\end{equation}
At $c=1$, $m_1(x,y)=(x+y)/2$, and the right-hand side is half
the relative-entropy variance, an identity also expressed
in terms of the refined R\'enyi divergence in
\cite[Appendix~B.1]{deBoer:2020hypothesis}. 
The differences between the kernels
at $c=1/2,1,2$ have a further interpretation.

For a Hermitian observable $A$, we use the conventions
\begin{align}
 F_Q(\rho,A)
 &\defeq2\sum_{i,j}\frac{(p_i-p_j)^2}{p_i+p_j}|A_{ij}|^2,
       \qquad \rho=\sum_i p_i|i\rangle\langle i|,
       \label{eq:sld-fisher-convention}\\
 I_{\rm WY}(\rho,A)
 &\defeq\Tr\rho A^2-\Tr\sqrt\rho A\sqrt\rho A.
       \label{eq:wy-convention}
\end{align}
Here $F_Q$ is the symmetric-logarithmic-derivative Fisher information
of the unitary family $e^{-isA}\rho e^{isA}$.

\begin{corollary}[Fisher and skew information from order slopes]
\label{cor:order-slopes-fisher}
For faithful density matrices,
\begin{align}
 S_{1/2}(\rho\Vert\sigma)-S_1(\rho\Vert\sigma)
 &=\frac18F_Q(\rho,\log\sigma),
       \label{eq:slope-qfi}\\
 S_1(\rho\Vert\sigma)-S_2(\rho\Vert\sigma)
 &=\frac14I_{\rm WY}(\rho,\log\sigma).
       \label{eq:slope-wy}
\end{align}
Both differences vanish precisely when $[\rho,\sigma]=0$.
\end{corollary}

\begin{proof}
The scalar kernels satisfy
\[
 m_{1/2}(x,y)-m_1(x,y)=\frac{(x-y)^2}{2(x+y)},\qquad
 m_1(x,y)-m_2(x,y)=\frac{(\sqrt x-\sqrt y)^2}{4}.
\]
Insert these identities into \eqref{eq:order-slope-variance}.
The centering terms cancel. In an eigenbasis of $\rho$, the
off-diagonal entries of $\log\rho-\log\sigma$ are
$-(\log\sigma)_{ij}$, and the two displayed kernel differences
vanish on the diagonal. Equations~\eqref{eq:slope-qfi} and
\eqref{eq:slope-wy} follow from the stated conventions. Their
vanishing is equivalent to $[\rho,\log\sigma]=0$, or,
equivalently, $[\rho,\sigma]=0$.
\end{proof}

The same calculation extends the Wigner--Yanase identity to
integer values of $c=n$. Define the Wigner--Yanase--Dyson skew
information by
\[
 I_s(\rho,A)\defeq\Tr\rho A^2
                   -\Tr\rho^sA\rho^{1-s}A,
 \qquad 0<s<1.
\]
The scalar identity
\begin{equation}
    m_n(x,y)=\frac{x+y}{2n} + \frac1n\sum_{k=1}^{n-1}x^{k/n}y^{1-k/n}
    \label{eq:scalar-identity}
\end{equation}
gives, for every integer $n\ge2$,
\begin{equation}
 S_1(\rho\Vert\sigma)-S_n(\rho\Vert\sigma)
 =\frac1{2n}\sum_{k=1}^{n-1}
     I_{k/n}(\rho,\log\sigma).
 \label{eq:integer-slope-wyd}
\end{equation}


For a Gibbs reference $\gamma_H=e^{-\beta_{\rm th}H}/Z$ with
$\beta_{\rm th}>0$, the identities in Corollary 
\ref{cor:order-slopes-fisher} become
\begin{equation}
 F_Q(\rho,H)=\frac8{\beta_{\rm th}^2}(S_{1/2}-S_1),
 \qquad
 I_{\rm WY}(\rho,H)=\frac4{\beta_{\rm th}^2}(S_1-S_2),
 \label{eq:thermal-slope-coherence}
\end{equation}
where each slope is evaluated at $(\rho,\gamma_H)$.
Thus the differences of the order derivatives measure coherence
between distinct energy eigenspaces. They vanish exactly when
$\rho$ commutes with $H$, irrespective of its energy populations.

Taking both states to be Gibbs
states at the same temperature, but with differing Hamiltonians, 
allows a response interpretation. Take
$\rho=e^{-\beta_{\rm th}H_0}/Z_0$ and
$\sigma=e^{-\beta_{\rm th}(H_0+V)}/Z_1$, and put
$\delta V=V-(\Tr\rho V)\one$. For an integer $c=n$, 
the finite-sum formula \eqref{eq:scalar-identity} for $m_n$ gives
\[
 \mathcal M_n^\rho(A)
 =\frac{\rho A+A\rho}{2n}
  +\frac1n\sum_{k=1}^{n-1}\rho^{k/n}A\rho^{1-k/n}.
\]
Applying this identity to the centered observable $\delta V$ and
using cyclicity of the trace, the two endpoint contributions
combine to give
\[
 \mathcal V_n^\rho(V)
 =\frac1n\sum_{k=0}^{n-1}
   \Tr\rho^{k/n}\delta V\rho^{1-k/n}\delta V,
\]
leading to
\begin{equation}
 2S_n=\frac{\beta_{\rm th}^2}{n}
       \sum_{k=0}^{n-1}
       \Tr\rho^{k/n}\delta V\rho^{1-k/n}\delta V.
 \label{eq:thermal-response-slope-sum}
\end{equation}
The sum samples the equilibrium imaginary-time correlation function
at equally spaced points. In finite dimensions this function is
continuous, so the Riemann sums converge as $n\to\infty$. Their
limit is the Bogoliubov--Kubo--Mori covariance:
\begin{equation}
 S_\infty\defeq\lim_{n\to\infty}S_n
 =\frac{\beta_{\rm th}^2}{2}
   \int_0^1\Tr\rho^s\delta V\rho^{1-s}\delta V\dd s
 =\frac{\beta_{\rm th}}2\chi_V,
 \label{eq:thermal-response-slope-limit}
\end{equation}
where
\[
 \chi_V\defeq-
 \left.\frac{\dd}{\dd\lambda}
 \frac{\Tr[V e^{-\beta_{\rm th}(H_0+\lambda V)}]}
      {\Tr e^{-\beta_{\rm th}(H_0+\lambda V)}}
 \right|_{\lambda=0}.
\]
The last equality follows from the Duhamel formula for the derivative
of the Gibbs state. At the other endpoint,
$S_1=\beta_{\rm th}^2\operatorname{Var}_\rho(V)/2$. The integer
slopes therefore connect the ordinary variance to the static
susceptibility through discrete imaginary-time averages. The two
endpoint fluctuation quantities also appear in
\cite{Jaksic:2011entropic}. Their difference is
\[
 S_1-S_\infty=\frac{\beta_{\rm th}^2}{2}
 \left[\operatorname{Var}_\rho(V)-\frac{\chi_V}{\beta_{\rm th}}\right].
\]
The expression in brackets is the quantum variance introduced in
\cite{Frerot:2015variance}.

\subsection{Escort averages and stability}

The integral representation separates the distinguishability of the
average escort state from the variation along the escort trajectory.
For $\alpha>0$, $\alpha\ne1$, let $I_\alpha$ be the interval joining
$1$ and $\alpha$, and let $\mu_\alpha$ be the probability measure
appearing in the integral representation. Define the averaged escort state
\begin{equation}
    \overline\rho_{\alpha,c}
 =\int_{I_\alpha}\rho_\beta^{(c)}\,\dd\mu_\alpha(\beta).
\end{equation}

For faithful density matrices, the Donald identity gives
\begin{equation}\label{eq:donald-holevo}
 D_{\alpha,c\alpha}(\rho\Vert\sigma)
 =D(\overline\rho_{\alpha,c}\Vert\sigma)+\chi_{\alpha,c},
\end{equation}
where
\[
 \chi_{\alpha,c}
 =\int_{I_\alpha}
 D(\rho_\beta^{(c)}\Vert\overline\rho_{\alpha,c})
 \,\dd\mu_\alpha(\beta).
\]
Writing $H(\tau)=-\Tr\tau\log\tau$, one also has
\[
 \chi_{\alpha,c}
 =H(\overline\rho_{\alpha,c})
  -\int_{I_\alpha}H(\rho_\beta^{(c)})
    \,\dd\mu_\alpha(\beta).
\]
Thus $\chi_{\alpha,c}$ is the Holevo information of the escort
ensemble, equivalently the increase in entropy caused by discarding
the order label. The decomposition follows directly by adding and
subtracting
$\Tr\rho_\beta^{(c)}\log\overline\rho_{\alpha,c}$
inside the integral representation.

In particular,
$D_{\alpha,c\alpha}(\rho\Vert\sigma)\geq
D(\overline\rho_{\alpha,c}\Vert\sigma)$.
Pinsker's inequality combined with equation~\eqref{eq:donald-holevo} and the integral representation quantifies the gap:
\[
 \int_{I_\alpha}
 \|\rho_\beta^{(c)}-\overline\rho_{\alpha,c}\|_1^2
 \,\dd\mu_\alpha(\beta)
 \leq 2\chi_{\alpha,c}.
\]
Near equality therefore means that the escort states have small
mean-square trace-norm dispersion around their average. Applying
Pinsker's inequality directly to the representation also gives
\[
 \int_{I_\alpha}
 \|\rho_\beta^{(c)}-\sigma\|_1^2
 \,\dd\mu_\alpha(\beta)
 \leq 2D_{\alpha,c\alpha}(\rho\Vert\sigma).
\]
A small divergence thus controls the average distance to the
reference, while a small Holevo term controls the variation within
the escort ensemble. 

Appendix~\ref{app:escort-stability} gives the corresponding
general-algebra bounds and the measurability conditions needed to
form a normal-state average.

\section{Optimized conditional entropies}
\label{sec:optimized-entropies}

Conditional entropies describe uncertainty about a system $A$
when one also has access to a correlated system $B$.
Their R\'enyi generalizations provide a finer description of
this uncertainty. In particular, optimized sandwiched conditional
entropies arise in distinguishing a joint state from a family
of uncorrelated reference states, where they determine how
testing probabilities scale with the number of copies
\cite{Hayashi:2014correlation}. We seek to express these
optimized quantities as averages of ordinary conditional
entropies along suitable escort families.

An optimized conditional entropy is obtained by minimizing a
divergence over reference states on $B$. Since our integral
representation holds for a fixed reference state, this
minimization cannot simply be moved inside the integral.
We instead differentiate the optimized value and identify its
derivative through the escort associated with a minimizing
reference state at that order. The key observation is that
this reference state coincides with the $B$-marginal of its
escort, turning the escort relative entropy into minus its
conditional entropy. This marginal characterization is
established for sandwiched divergences in
\cite{Hayashi:2014correlation} and for generalized conditional
entropies in \cite{Rubboli:2024conditional}.
We give the finite-dimensional argument needed to obtain the
integral representation, without assuming uniqueness of the
optimizer.

\subsection{The conditional-entropy integral}

Throughout this section, $\rho_{AB}>0$ is a density matrix on a
finite-dimensional bipartite system, and $c>0$. We write
$H(\tau)=-\Tr\tau\log\tau$ and
$H(A|B)_\tau=H(\tau_{AB})-H(\tau_B)$.
We use the usual
extension to the unnormalized comparison operator
$\one_A\otimes\sigma_B$:
\[
 D_{\alpha,z}(\rho_{AB}\Vert\one_A\otimes\sigma_B)
 =D_{\alpha,z}\!\left(
 \rho_{AB}\middle\Vert\frac{\one_A}{d_A}\otimes\sigma_B\right)
 -\log d_A.
\]
For $\alpha>1$, define
\begin{equation}
 J_c(\alpha)\defeq
 \inf_{\sigma_B}D_{\alpha,c\alpha}
       (\rho_{AB}\Vert\one_A\otimes\sigma_B),
 \qquad
 H^\uparrow_{\alpha,c\alpha}(A|B)_\rho\defeq-J_c(\alpha),
 \label{eq:optimized-conditional-definition}
\end{equation}
where $\sigma_B$ ranges over density matrices. Put
$\kappa=1-\alpha^{-1}$ and $G_c(\kappa)=\kappa J_c(\alpha)$.
For a faithful comparison state $\sigma_B$, write
\begin{equation}
 X_{\alpha,c}(\sigma_B)
 =(\one_A\otimes\sigma_B)^{-\kappa/(2c)}
       \rho_{AB}^{1/c}
  (\one_A\otimes\sigma_B)^{-\kappa/(2c)},
 \qquad
 \eta_{\alpha,c}(\sigma_B)
 =\frac{X_{\alpha,c}(\sigma_B)^{c\alpha}}
        {\Tr X_{\alpha,c}(\sigma_B)^{c\alpha}}.
 \label{eq:optimized-escort-definition}
\end{equation}
The scalar normalization of the comparison operator does not change
this escort. Writing $f_c(\kappa,\sigma_B)=\kappa D_{\alpha,c\alpha} 
(\rho_{AB}\|\one_A\otimes\sigma_B)$, the optimized potential is
\[
    G_c(\kappa) = \inf_{\sigma_B}f_c(\kappa,\sigma_B) = \kappa J_c(\alpha).
\]
In terms of $\kappa$, the fixed-reference identity 
\eqref{eq:fd-ray-refined-identity} becomes
\begin{equation}
  \partial_\kappa f_c(\kappa,\sigma_B) 
  = D(\eta_{\alpha,c}(\sigma_B) \|\one_A\otimes\sigma_B).
\end{equation}
The next proposition establishes what happens to this identity 
after optimization, and expresses the optimized R\'enyi
conditional entropy as a weighted average of ordinary 
conditional entropies.

\begin{proposition}[Optimized conditional-entropy representation]
\label{prop:optimized-conditional-integral}
For every $c>0$ and $\alpha>1$, the infimum in
\eqref{eq:optimized-conditional-definition} is attained at a faithful
state. Every optimizer $\sigma_{\alpha,c}$ satisfies
\begin{equation}
 \Tr_A\eta_{\alpha,c}(\sigma_{\alpha,c})=\sigma_{\alpha,c}.
 \label{eq:optimized-escort-fixed-point}
\end{equation}
The function $G_c$ extends to $G_c(0)=0$ and is
$\log d_A$-Lipschitz on $[0,1)$. At almost every $\kappa\in(0,1)$,
\begin{equation}
 G_c'(\kappa)=-H(A|B)_{\eta_{\alpha,c}(\sigma_{\alpha,c})}.
 \label{eq:optimized-envelope-derivative}
\end{equation}
At every differentiability point, the right-hand side is the same
for all optimizers. Consequently,
\begin{equation}
 H^\uparrow_{\alpha,c\alpha}(A|B)_\rho
 =\frac{\alpha}{\alpha-1}\int_1^\alpha
   \frac{H(A|B)_{\eta_{\beta,c}(\sigma_{\beta,c})}}{\beta^2}
   \,\dd\beta.
 \label{eq:optimized-conditional-integral}
\end{equation}
\end{proposition}

\begin{proof}
Fixed-ray monotonicity \eqref{eq:upper-endpoint-bounds} applied 
to the normalized comparison state gives
\begin{equation}
 D_{\alpha,c\alpha}(\rho_{AB}\Vert\one_A\otimes\sigma_B)
 \geq D(\rho_{AB}\Vert\one_A\otimes\sigma_B)
 =-H(A|B)_\rho+D(\rho_B\Vert\sigma_B).
 \label{eq:optimized-coercivity}
\end{equation}
Since $\rho_B>0$, the last term tends to infinity when
$\sigma_B$ approaches the boundary of the state space. On the other hand, the choice $\sigma_B=\rho_B$ has finite divergence. Continuity in the interior
therefore proves existence of a faithful minimizer. On every compact
order interval, the same choice $\sigma_B=\rho_B$ bounds all minimizing values
uniformly. Equation~\eqref{eq:optimized-coercivity} then places all
minimizers in one compact subset of the faithful states.

We next establish the marginal identity 
\eqref{eq:optimized-escort-fixed-point}. 
Fix an order and put
$r=-\kappa/(2c)$. In an eigenbasis of
$\sigma_B=\sum_i s_i|i\rangle\langle i|$, define
\[
 [\mathcal K_r^{\sigma_B}(Z)]_{ij}
 =\frac{\sinh(r\log(s_i/s_j))}{r(s_i-s_j)}\,Z_{ij},
 \qquad
 [\mathcal K_r^{\sigma_B}(Z)]_{ii}=\frac{Z_{ii}}{s_i},
\]
with the continuous value for coincident eigenvalues.
Differentiating the matrix powers in the comparison state gives,
for a Hermitian variation $K_B$,
\begin{equation}
 \partial D_{\alpha,c\alpha}(\sigma_B)[K_B]
 =-\Tr K_B\,
   \mathcal K_r^{\sigma_B}\!
       \bigl(\Tr_A\eta_{\alpha,c}(\sigma_B)\bigr),
 \label{eq:optimized-comparison-gradient}
\end{equation}
where $\partial D_{\alpha,c\alpha}$ is the Fr\'echet derivative.
At an interior minimizer, this must vanish for every traceless
Hermitian matrix $K_B$. $\mathcal K_r^{\sigma_B}(\Tr_A\eta_{\alpha,c}(\sigma_B))$ is
thus proportional to the identity matrix.
The definition of $\mathcal K_r^{\sigma_B}$ then forces the
escort marginal to be proportional to
$\sigma_B$. Both have trace one, proving
\eqref{eq:optimized-escort-fixed-point}.

The functional $f_c(\kappa,\sigma_B)=
\kappa D_{\alpha,c\alpha}(\rho_{AB}\Vert\one_A\otimes\sigma_B)$
is smooth for $0<\kappa<1$ and $\sigma_B>0$.
The uniform compactness just proved bounds its order derivative on
every compact parameter interval, so its minimum $G_c$ is locally
Lipschitz. At a point where $G_c$ is differentiable, choose any
optimizer $\sigma_{\alpha,c}$. The difference
$f_c(\,\cdot\,,\sigma_{\alpha,c})-G_c$ has a minimum of zero
there, and hence their derivatives agree at the chosen optimizer. 
Thus, at every differentiability point of $G_c$,
\begin{equation}
 G_c'(\kappa)
 =\left.
   \partial_\kappa f_c(\kappa,\sigma_B)
   \right|_{\sigma_B=\sigma_{\alpha,c}}.
 \label{eq:optimized-envelope-identity}
\end{equation}
The fixed-reference refined identity~\eqref{eq:fd-ray-refined-identity} then gives
\[
 \partial_\kappa f_c
 =D(\eta_{\alpha,c}\Vert\one_A\otimes\sigma_{\alpha,c})
 =-H(A|B)_{\eta_{\alpha,c}},
\]
where the last equality uses
\eqref{eq:optimized-escort-fixed-point}.
This proves \eqref{eq:optimized-envelope-derivative}, including its
independence of the optimizer.

The bound $|H(A|B)_\eta|\leq\log d_A$ shows that $G_c$ is
$\log d_A$-Lipschitz between any two positive parameters. The lower
bound \eqref{eq:optimized-coercivity} and the choice $\sigma_B=\rho_B$
also imply $G_c(\kappa)\to0$ as $\kappa\downarrow0$.
Integration from zero, followed by
$\dd\kappa=\beta^{-2}\dd\beta$, proves
\eqref{eq:optimized-conditional-integral}.
\end{proof}

At $c=1$, this is the conditional-entropy identity of
\cite[Lemma~7]{Cheng:2026transmission}. 

\subsection{Order-one derivatives and the response of the optimizer}

At order one, the minimizing comparison state is $\rho_B$.
Optimization also leaves the first order derivative at $\alpha=1_+$ 
equal to the fixed-reference derivative at this state.

\begin{lemma}[Optimized order-one slopes]
\label{lem:optimized-slopes}
Let
\[
 \mathsf S_c\defeq
 \left.\frac{\dd}{\dd\alpha}J_c(\alpha)\right|_{\alpha=1+},
 \qquad
 L_0=\log\rho_{AB}-\log(\one_A\otimes\rho_B).
\]
The right derivative exists for every $c>0$ and satisfies
\begin{equation}
 \mathsf S_c=\frac12\mathcal V_c^{\rho_{AB}}(L_0).
 \label{eq:optimized-order-one-slope}
\end{equation}
In particular, with $A_B=\one_A\otimes\log\rho_B$,
\begin{equation}
 \mathsf S_{1/2}-\mathsf S_1=\frac18F_Q(\rho_{AB},A_B),
 \qquad
 \mathsf S_1-\mathsf S_2=\frac14I_{\rm WY}(\rho_{AB},A_B).
 \label{eq:optimized-slope-information}
\end{equation}
The corresponding derivatives of $H^\uparrow_{\alpha,c\alpha}$
have the opposite sign.
\end{lemma}

\begin{proof}
Write
\[
 F_\alpha(\sigma_B)
 =D_{\alpha,c\alpha}(\rho_{AB}\Vert\one_A\otimes\sigma_B).
\]
For any optimizer $\sigma_{\alpha,c}$,
\eqref{eq:optimized-coercivity} and the trial state $\rho_B$ give
\[
 0\leq D(\rho_B\Vert\sigma_{\alpha,c})
 \leq F_\alpha(\rho_B)+H(A|B)_\rho
 \longrightarrow0
 \qquad(\alpha\downarrow1).
\]
Thus $\sigma_{\alpha,c}\to\rho_B$ in trace norm, independently of
the choice of optimizer.

Since $\rho_B>0$, all optimizers sufficiently close to order one
lie in a fixed compact neighborhood consisting of faithful states.
The smooth removable extension of $F_\alpha$ at $\alpha=1$
therefore gives the uniform expansion
\[
 F_{1+\varepsilon}(\sigma_B)
 =-H(A|B)_\rho+D(\rho_B\Vert\sigma_B)
   +\varepsilon a_c(\sigma_B)+O(\varepsilon^2),
 \qquad
 a_c(\sigma_B)
 =\left.\partial_\alpha F_\alpha(\sigma_B)\right|_{\alpha=1},
\]
where $a_c$ is continuous.

Evaluating at $\rho_B$ gives an upper bound on $J_c(1+\varepsilon)$.
Evaluating at an optimizer and using
$D(\rho_B\Vert\sigma_{1+\varepsilon,c})\geq0$ gives the lower bound
\[
 a_c(\sigma_{1+\varepsilon,c})-C\varepsilon
 \leq
 \frac{J_c(1+\varepsilon)+H(A|B)_\rho}{\varepsilon}
 \leq a_c(\rho_B)+C\varepsilon
\]
for some constant $C$. Since
$a_c(\sigma_{1+\varepsilon,c})\to a_c(\rho_B)$, the right derivative
exists and equals $a_c(\rho_B)$.

Finally, \eqref{eq:order-slope-variance} yields
\[
 \mathsf S_c=a_c(\rho_B)
 =\frac12\mathcal V_c^{\rho_{AB}}(L_0).
\]
To apply that identity with a normalized reference, replace
$\one_A\otimes\rho_B$ by $(\one_A/d_A)\otimes\rho_B$.
This changes the divergence by an order-independent constant
and adds a scalar multiple of the identity to $L_0$.
Neither change affects the claimed formula.

 Corollary~\ref{cor:order-slopes-fisher} then proves
\eqref{eq:optimized-slope-information}.
\end{proof}

Away from order one, the motion of the optimizer contributes to
curvature. This contribution can be stated precisely when the
minimizer is nondegenerate.

\begin{proposition}[Curvature along a smooth optimizing branch]
\label{prop:optimized-curvature}
Let
\[
 \mathcal T_B=\{K=K^\dagger:\Tr K=0\},
 \qquad
 d_c(\kappa,\sigma)
 =D_{\alpha,c\alpha}
   (\rho_{AB}\Vert\one_A\otimes\sigma),
 \qquad
 \alpha=(1-\kappa)^{-1}.
\]
Write $\partial$ for Fr\'echet differentiation in $\sigma$,
and $\partial_\kappa$ for differentiation at fixed $\sigma$.
Suppose that a minimizing comparison state $\sigma(\kappa)$
is smooth on an open subinterval of $(0,1)$.
Define the Hessian operator
$\mathcal H_\kappa:\mathcal T_B\to\mathcal T_B$
and $v_\kappa\in\mathcal T_B$ by
\begin{align*}
 \Tr\!\left[K\,\mathcal H_\kappa[L]\right]
 &=\partial^2d_c(\kappa,\sigma(\kappa))[K,L],\\
 \Tr(Kv_\kappa)
 &=\left.
   \partial_\kappa\bigl(\partial d_c(\kappa,\sigma)[K]\bigr)
   \right|_{\sigma=\sigma(\kappa)},
 \qquad K,L\in\mathcal T_B.
\end{align*}
Assume that $\mathcal H_\kappa$ is positive definite.
Writing $\sigma=\sigma(\kappa)$ and
$\eta=\eta_{\alpha,c}(\sigma)$, we have
\begin{equation}
 G_c''(\kappa)=
 \alpha\mathcal V_{c\alpha}^{\eta}
 \bigl(\log\eta-\log(\one_A\otimes\sigma)\bigr)
 -\kappa\Tr\!\left[
 v_\kappa\,\mathcal H_\kappa^{-1}[v_\kappa]
 \right].
 \label{eq:optimized-curvature-correction}
\end{equation}
\end{proposition}

\begin{proof}
All derivatives below are evaluated at
$\sigma=\sigma(\kappa)$ unless indicated otherwise.
Stationarity gives $\partial d_c[K]=0$ for every
$K\in\mathcal T_B$. Differentiating along the minimizing branch,
with $K$ fixed, yields
\[
 0=\partial_\kappa(\partial d_c[K])
   +\partial^2d_c[K,\sigma']
  =\Tr\!\left[K\bigl(v_\kappa+
                     \mathcal H_\kappa[\sigma']\bigr)\right].
\]
Since this holds for every $K\in\mathcal T_B$, and
$\sigma'\in\mathcal T_B$, we obtain
\[
 v_\kappa+\mathcal H_\kappa[\sigma']=0,
 \qquad
 \sigma'=-\mathcal H_\kappa^{-1}[v_\kappa].
\]
The inverse exists by positive definiteness.

Set $f_c(\kappa,\sigma)=\kappa d_c(\kappa,\sigma)$.
Since $G_c(\kappa)=f_c(\kappa,\sigma(\kappa))$, stationarity
and the chain rule give
\begin{equation}
 G_c'=\partial_\kappa f_c+\partial f_c[\sigma']
     =\partial_\kappa f_c.
\end{equation}
Moreover, at the optimizer,
\[
 \partial(\partial_\kappa f_c)[K]
 =\partial d_c[K]
  +\kappa\partial_\kappa(\partial d_c[K])
 =\kappa\Tr(Kv_\kappa).
\]
Differentiating once more therefore gives
\begin{align*}
 G_c''
 &=\partial_\kappa^2f_c
   +\partial(\partial_\kappa f_c)[\sigma']\\
 &=\partial_\kappa^2f_c+\kappa\Tr(\sigma'v_\kappa)\\
 &=\partial_\kappa^2f_c
   -\kappa\Tr\!\left[
      v_\kappa\,\mathcal H_\kappa^{-1}[v_\kappa]
     \right].
\end{align*}
Finally, the fixed-reference identity
\eqref{eq:convex-potential-susceptibility} gives
\[
 \partial_\kappa^2f_c
 =\alpha\mathcal V_{c\alpha}^{\eta}
   \bigl(\log\eta-\log(\one_A\otimes\sigma)\bigr),
\]
which proves the result.
Normalizing the comparison operator changes $f_c$ only by
a linear function of $\kappa$, so it does not affect this
second-derivative identity.
\end{proof}

The second term in \eqref{eq:optimized-curvature-correction}
is the correction due to changing the minimizing state. The
fixed-reference susceptibility alone therefore does not give the
curvature of the optimized value.

\subsection{Classical--quantum states with fixed input distribution}

Let $p_x>0$ be a fixed distribution on a finite set, let
$\rho_B^x>0$ be density matrices, and set
$\rho_{XB}=\sum_xp_x|x\rangle\langle x|\otimes\rho_B^x$ and
$p_X=\sum_xp_x|x\rangle\langle x|$. Define
\begin{equation}
 \mathcal I_c(\alpha;p)
 =\inf_{\sigma_B}D_{\alpha,c\alpha}
        (\rho_{XB}\Vert p_X\otimes\sigma_B).
 \label{eq:optimized-cq-definition}
\end{equation}
The escort $\eta_{\beta,c}$ below is formed from this pair using
an optimizer at order $\beta$. We use
$I(X:B)_\eta=D(\eta_{XB}\Vert\eta_X\otimes\eta_B)$
for ordinary mutual information.

\begin{proposition}[Integral with fixed classical distribution]
\label{prop:optimized-cq-integral}
For every $c>0$ and $\alpha>1$,
\begin{equation}
 \mathcal I_c(\alpha;p)
 =\frac{\alpha}{\alpha-1}\int_1^\alpha
 \frac{I(X:B)_{\eta_{\beta,c}}
       +D((\eta_{\beta,c})_X\Vert p_X)}{\beta^2}\,\dd\beta.
 \label{eq:optimized-cq-integral}
\end{equation}
As in Proposition~\ref{prop:optimized-conditional-integral}, the
scalar integrand is independent of the optimizer almost everywhere.
It lies between zero and $\log(1/p_{\min})$.
\end{proposition}

\begin{proof}
The compactness and envelope argument of
Proposition~\ref{prop:optimized-conditional-integral} applies with
the order-one lower bound
$I(X:B)_\rho+D(\rho_B\Vert\sigma_B)$.
Write $\eta=\eta_{\alpha,c}$ for an optimizing escort.
Stationarity in $\sigma_B$ again gives $\eta_B=\sigma_B$, so,
with $\alpha=(1-\kappa)^{-1}$, the envelope identity~\eqref{eq:optimized-envelope-identity} yields
\[
 \frac{\dd}{\dd\kappa}
 \bigl[\kappa\mathcal I_c((1-\kappa)^{-1};p)\bigr]
 =D(\eta_{XB}\Vert p_X\otimes\eta_B)
\]
almost everywhere. Adding and subtracting
$\log(\eta_X\otimes\eta_B)$ gives
\[
 D(\eta_{XB}\Vert p_X\otimes\eta_B)
 =I(X:B)_\eta+D(\eta_X\Vert p_X).
\]
The second term remains because the classical reference $p_X$
is fixed: stationarity imposes $\eta_B=\sigma_B$, but does not
impose $\eta_X=p_X$.

Writing $q_x=\eta_X(x)$ and using
$H(X|B)_\eta\geq0$ for classical--quantum states, we obtain
\[
 0\leq I(X:B)_\eta+D(\eta_X\Vert p_X)
 =-H(X|B)_\eta-\sum_xq_x\log p_x
 \leq\log(1/p_{\min}).
\]
The preceding local absolute-continuity argument therefore
permits integration from $\kappa=0$, where the optimized
potential vanishes. Substituting
$t=1-\beta^{-1}$ proves \eqref{eq:optimized-cq-integral}.
Since every optimizer gives the same envelope derivative
where it exists, the scalar integrand is independent of the
optimizer almost everywhere.
\end{proof}

If the input distribution is also optimized, and the maximization
over $p$ and minimization over $\sigma_B$ admit a saddle point,
stationarity implies $\eta_X=p_X$ at that pair. The correction
$D(\eta_X\Vert p_X)$ then vanishes, leaving only the mutual
information in the integrand. This occurs in the sandwiched
Holevo-information representation of \cite{Cheng:2026transmission},
where $p$ is optimized separately at each order.
For a fixed input distribution, as in
Proposition~\ref{prop:optimized-cq-integral}, the two classical
marginals need not agree, so the correction must be retained.

\section{Applications}
\label{sec:applications}

We consider applications of our integral representation for coherent states in quantum field theory, hypothesis testing,
and work extraction. The first application evaluates divergences from
an explicit escort relative entropy. In the operational applications,
we use established exponent formulas and derive their escort-integral
and response forms. The reference state in each optimization is prescribed.

\subsection{Coherent states and null deformations}
\label{subsec:qft-applications}

The coherent-state construction of \cite{Kibe:2026escorts} gives
explicit escort profiles on a null fibre. We recall the profiles and
evaluate the order integral that determines their R\'enyi divergences.

Consider the chiral $U(1)$ current $J(x)$, realized as the derivative
of a free scalar field along a null fibre. The half-line algebra
$M_V=\mathcal A((V,\infty))$ describes observables to the future of
the cut $V$. We use the normalization $T_{++}=\tfrac12:J^2:$, where
normal ordering is with respect to the vacuum.
A coherent excitation changes the mean current and preserves the
connected vacuum two-point function. Choose a state $\psi$
such that its current expectation agrees with
\begin{equation}
 j_\psi(x)=\frac{qg}{(x+g)^2},\qquad x>-g/2,
 \qquad g>0,\quad q\in\mathbb R\setminus\{0\}.
 \label{eq:application-coherent-current}
\end{equation}
A smooth completion to the left, with a compensating pulse, makes the
total current integral vanish without changing the half-line
restrictions considered here. For $0\leq V<g/4$, put $g_V=g+V$ and
write $\psi_V,\omega_V$ for the excited and vacuum states on $M_V$.
In the common-ultraviolet-regulator construction of
\cite{Kibe:2026escorts}, the sandwiched escorts are again coherent.
Their continuum current profiles, for $x>0$ and $\beta\geq1$, are
\begin{equation}
 j_{\beta,V}(V+x)
 =\frac{qg\beta^2}{g_V^2}
   \frac{(x/g_V)^{\beta-1}}
        {[1+(x/g_V)^\beta]^2}.
 \label{eq:application-coherent-escort-profile}
\end{equation}
The half-line relative entropy of a coherent excitation is its
vacuum modular energy \cite[Theorem~4.7]{Longo:2018distribution}, \cite{Morinelli:2021modular}. Hence
\begin{equation}
 D(\psi_{\beta,V}^s\Vert\omega_V)
 =\pi\int_0^\infty x\,j_{\beta,V}(V+x)^2\,\dd x
 =\frac{\pi q^2g^2}{6g_V^2}\,\beta^3.
 \label{eq:application-coherent-escort-entropy}
\end{equation}
The last equality follows by setting $t=x/g_V$ and using
$\int_0^\infty t^{2\beta-1}(1+t^\beta)^{-4}\,\dd t=1/(6\beta)$.
Applying our integral representation gives, for every $\alpha>1$,
\begin{align}
 D_\alpha(\psi_V\Vert\omega_V)
 &=\frac{\alpha}{\alpha-1}\int_1^\alpha
   \frac{D(\psi_{\beta,V}^s\Vert\omega_V)}{\beta^2}\,\dd\beta
 \notag\\
 &=\frac{\pi q^2g^2}{12g_V^2}\,\alpha(\alpha+1),
 \label{eq:application-coherent-srd}\\
 \frac{\dd^2}{\dd V^2}D_\alpha(\psi_V\Vert\omega_V)
 &=\frac{\pi q^2g^2}{2g_V^4}\,\alpha(\alpha+1)>0.
 \label{eq:application-coherent-srd-hessian}
\end{align}
The divergence is therefore finite and strictly convex in the cut
for every $\alpha>1$.

For the two-parameter $\alpha$--$z$ family, the fixed-ray escort calculation in
\cite{Kibe:2026escorts} and the order integral give
\begin{align}
 D_{\alpha,z}(\psi_V\Vert\omega_V)
 &=\frac{q^2g^2}{g_V^2}\,\mathcal A_{\alpha,z},
 \label{eq:application-coherent-az}\\
 \mathcal A_{\alpha,z}
 &=\frac{z}{\pi(\alpha-1)}\int_0^\infty
 \frac{s\sinh(\alpha s/z)\sinh((\alpha-1)s/z)}
      {\sinh^2s\,\sinh(s/z)}\,\dd s,
 \qquad\alpha>1.
 \label{eq:application-coherent-coefficient}
\end{align}
The integrand is positive and has large-$s$ behavior
$2s\exp[-2(z-\alpha+1)s/z]$. The divergence is therefore
finite exactly when $z>\alpha-1$.
In the finite range,
\begin{equation}
 \frac{\dd^2}{\dd V^2}D_{\alpha,z}(\psi_V\Vert\omega_V)
 =\frac{6q^2g^2}{g_V^4}\,\mathcal A_{\alpha,z}>0.
 \label{eq:application-coherent-az-hessian}
\end{equation}
Thus the $\alpha$--$z$ divergence is convex in the cut when it is finite.

The same calculation supplies a finite nonconstant one-null-fibre
excitation for the off-diagonal focusing obstruction of
\cite{Kibe:2026focusing}. For two cuts $0\leq u<w<g/4$, put
$\delta_\alpha=D_\alpha(\psi_u\Vert\omega_u)
                -D_\alpha(\psi_w\Vert\omega_w)>0$.
Conditioning exclusive activation of two identical null fibres on a shared
classical record of weights $p$ and $1-p$ gives
\begin{equation}
 \Delta_{12}D_\alpha
 =-\frac1{\alpha-1}\log\!\left[
 1+4p(1-p)\sinh^2\!\left(\frac{\alpha-1}{2}\delta_\alpha\right)
 \right]<0,
 \qquad0<p<1.
 \label{eq:application-coherent-offdiagonal}
\end{equation}
Here $\Delta_{12}D=D(w,w)+D(u,u)-D(u,w)-D(w,u)$.
This provides an explicit violation of a potential 
off-diagonal focusing inequality \cite{Kibe:2026focusing}.

The coherent-state divergences above can also be evaluated
directly, without using the integral representation. Its
principal value here lies in making their physical content
more transparent through the energy interpretation of relative
entropy. Under the ultraviolet regularization and limiting
assumptions detailed in \cite{Kibe:2026escorts}, this connection
allows the R\'enyi quantum null energy condition (QNEC) \cite{Lashkari:2018constraining} to be reformulated as a bound on the
energy of an appropriate escort state, similar to how the ordinary QNEC is a bound on null energy. This physical
reformulation was the central motivation of \cite{Kibe:2026escorts}.

The endpoint escort is also useful for the R\'enyi ant
formula of \cite{Kibe:2026integer}, which generalizes the ant formula for relative entropy conjectured by \cite{Wall:2017lower} and proved by \cite{Ceyhan:2018recovering,Hollands:2025qnec}. 
Let $M_V$ now be the decreasing family
of algebras obtained from a half-sided modular inclusion by its
canonical translations. For $\alpha>1$, \cite{Kibe:2026integer} conjectured the 
R\'enyi ant formula
\begin{equation}
 -\frac{1}{\kappa_\alpha}\frac{\dd}{\dd V}D_\alpha(\psi\Vert\omega;M_V)
 =\inf_{\zeta\in\mathcal C(\varphi_{\alpha,V})}
   \|P^{1/2}\zeta\|^2,
 \qquad
 \kappa_\alpha=\frac{2\alpha\sin(\pi/\alpha)}{\alpha-1}.
 \label{eq:application-escort-ant}
\end{equation}
The infimum runs over vector purifications of the sandwiched escort $\varphi_{\alpha,V}$. The operator
$P\geq0$ generates the canonical null translations of a half-sided modular inclusion, which corresponds to the smeared averaged null energy for algebras associated to sufficiently symmetric regions such as constant cross sections of a null plane. For every integer
$n\geq2$, this R\'enyi ant formula has been proved \cite{Kibe:2026integer}.
The integral representation supplies
\eqref{eq:application-coherent-srd}, used in \cite{Kibe:2026integer} to verify the ant formula for this coherent family at every real $\alpha>1$.

\subsection{Hypothesis testing and the optimizing order}
\label{subsec:binary-testing}

The strong-converse exponent of asymmetric binary hypothesis
testing is given by an optimization over the sandwiched
R\'enyi order. We apply our integral representation
to this established variational formula and characterize
its interior optimizers through ordinary relative entropy.

Let $0\leq T\leq\one$ be a binary test on $M$, where the
outcome associated with $T$ accepts $\psi$. Its success
probability under $\psi$ and type-II error under $\omega$ are
\[
 p=\psi(T),
 \qquad
 q=\omega(T).
\]
For $\alpha>1$ and $(\alpha,z)$ in the data-processing region
\eqref{eq:dpi-region}, assume that
$D_{\alpha,z}(\psi\Vert\omega)<\infty$.
Monotonicity under measurement of $T$ gives
\begin{equation*}
 -\log p
 \geq
 \frac{\alpha-1}{\alpha}
 \bigl[-\log q-D_{\alpha,z}(\psi\Vert\omega)\bigr].
\end{equation*}
%
For tests on $n$ independent copies satisfying
$q_n\leq\mathrm e^{-nr}$, tensor additivity turns
this into a lower bound
\[
 \liminf_{n\to\infty}-\frac1n\log p_n
 \geq
 \frac{\alpha-1}{\alpha}
 \bigl[r-D_{\alpha,z}(\psi\Vert\omega)\bigr].
\]
On the sandwiched ray $z=\alpha$, this bound holds
at every order $\alpha>1$ with finite divergence.
Optimizing over $\alpha$ therefore gives the strongest
bound supplied by this family.

To recall the exact operational result, define
\begin{align}
 p_n^\star(r)
 &=
 \sup\bigl\{
 \psi^{\otimes n}(T):
 0\leq T\leq\one,\ 
 \omega^{\otimes n}(T)\leq\mathrm e^{-nr}
 \bigr\},
 \notag\\
 \scal(r)
 &=
 -\lim_{n\to\infty}\frac1n\log p_n^\star(r).
 \label{eq:optimal-testing-success}
\end{align}
Under the hypotheses stated in
Corollary~\ref{cor:testing-escort-optimization} below,
the limit exists and, for $r>D(\psi\Vert\omega)$,
the established strong-converse theorem gives
\begin{equation}
 \scal(r)
 =
 \sup_{\substack{\alpha>1\\
                  Q_\alpha(\psi\|\omega)<\infty}}
 \frac{\alpha-1}{\alpha}
 \bigl[r-D_\alpha(\psi\Vert\omega)\bigr].
 \label{eq:known-strong-converse}
\end{equation}
This operational characterization was proved in finite
dimensions by \cite{Mosonyi:2013testing}, extended to
injective von Neumann algebras by
\cite{Hiai:2021discrimination}, and established under
bounded domination on arbitrary von Neumann algebras
by \cite{Junge:2025converse}.

The following corollary expresses the exponent and its
optimizing order through the escort relative-entropy
profile, and gives the corresponding integral
representation of generalized cutoff rates.

\begin{corollary}[Escort representation and optimizing order]
\label{cor:testing-escort-optimization}
Let $\psi,\omega$ be normal states on a von Neumann algebra $M$.
Assume that either
\begin{enumerate}[label=(\alph*),leftmargin=*]
\item $M$ is injective and
      $D_{\alpha_0}(\psi\Vert\omega)<\infty$
      for some $\alpha_0>1$; or
\item $a\psi\leq\omega$ for some $a>0$.
\end{enumerate}
Let $\scal(r)$ be the optimal testing exponent defined in
\eqref{eq:optimal-testing-success}, and recall
\[
 \mathcal J_1
 =
 \{\beta>0:0<Q_\beta(\psi\|\omega)<\infty\},
 \qquad
 R_1(\beta)=D(\psi_\beta^s\Vert\omega).
\]
Then the following statements hold.

\begin{enumerate}[label=(\roman*),leftmargin=*]
\item
For every $r>D(\psi\Vert\omega)$,
\begin{equation}
 \scal(r)
 =
 \sup_{\alpha\in\mathcal J_1\cap(1,\infty)}
 \int_1^\alpha
 \frac{r-R_1(\beta)}{\beta^2}\,\dd\beta.
 \label{eq:strong-converse-area}
\end{equation}

\item
Fix $r>D(\psi\Vert\omega)$.
An order
$\alpha_r\in\operatorname{int}(\mathcal J_1)\cap(1,\infty)$
attains the supremum in \eqref{eq:strong-converse-area}
if and only if
\begin{equation}
 R_1(\alpha_r-)\leq r\leq R_1(\alpha_r+),
 \label{eq:strong-converse-level-crossing}
\end{equation}
where the one-sided limits are taken within
$\operatorname{int}(\mathcal J_1)$.
At a continuity point of $R_1$, this condition becomes
\begin{equation}
 D(\psi_{\alpha_r}^s\Vert\omega)=r.
 \label{eq:testing-escort-crossing}
\end{equation}
If $R_1$ is strictly increasing on
$\mathcal J_1\cap(1,\infty)$, any such interior optimizer
is unique. The supremum may instead be approached at
the boundary of the finite-moment domain, including
$\alpha\to\infty$.

\item
Under bounded domination, for every $0<\kappa<1$ the
generalized cutoff rate
\[
 C_\kappa
 =
 \inf\bigl\{
 r_0\in\mathbb R:
 \scal(t)\geq\kappa(t-r_0)
 \text{ for every }t>0
 \bigr\}
\]
satisfies
\begin{equation}
 C_\kappa
 =
 D_{1/(1-\kappa)}(\psi\Vert\omega)
 =
 \frac1\kappa\int_0^\kappa
 R_1\!\left(\frac1{1-u}\right)\,\dd u.
 \label{eq:cutoff-rate-average}
\end{equation}
\end{enumerate}
\end{corollary}

\begin{proof}
Under the stated hypotheses, the established operational
theorems
\cite{Mosonyi:2013testing,Hiai:2021discrimination,
      Junge:2025converse}
give the existence of the limit defining $\scal(r)$ and
the variational formula
\eqref{eq:known-strong-converse}.

For every $\alpha\in\mathcal J_1\cap(1,\infty)$,
Theorem~\ref{thm:fixed-ray-integral}, specialized to the
sandwiched ray, gives
\[
 \frac{\alpha-1}{\alpha}
 D_\alpha(\psi\Vert\omega)
 =
 \int_1^\alpha\frac{R_1(\beta)}{\beta^2}\,\dd\beta.
\]
Since
\[
 \frac{\alpha-1}{\alpha}
 =
 \int_1^\alpha\frac{\dd\beta}{\beta^2},
\]
we obtain
\[
 \frac{\alpha-1}{\alpha}
 \bigl[r-D_\alpha(\psi\Vert\omega)\bigr]
 =
 \int_1^\alpha
 \frac{r-R_1(\beta)}{\beta^2}\,\dd\beta.
\]
Taking the supremum over $\alpha$ proves~(i).

For~(ii), put $\kappa=1-\alpha^{-1}$ and use the
potential $\Phi_1$ from
Section~\ref{subsec:order-convexity}. The objective is
\[
 \kappa r-\Phi_1(\kappa),
 \qquad
 \Phi_1(\kappa)
 =
 \int_0^\kappa
 R_1\!\left(\frac1{1-u}\right)\,\dd u.
\]
By Corollary~\ref{cor:raywise-monotonicity}, the integrand
is nondecreasing, so $\Phi_1$ is convex. At an interior
point $\kappa_r=1-\alpha_r^{-1}$, its one-sided
derivatives are
\[
 \Phi'_{1,-}(\kappa_r)=R_1(\alpha_r-),
 \qquad
 \Phi'_{1,+}(\kappa_r)=R_1(\alpha_r+).
\]
Consequently, $\kappa_r$ maximizes
$\kappa r-\Phi_1(\kappa)$ if and only if
\[
 r\in\partial\Phi_1(\kappa_r)
 =
 [R_1(\alpha_r-),R_1(\alpha_r+)].
\]
This proves \eqref{eq:strong-converse-level-crossing}.
Continuity gives \eqref{eq:testing-escort-crossing}.
If $R_1$ is strictly increasing, then $\Phi_1$ is
strictly convex on its positive domain, which proves
uniqueness of an interior optimizer.

Finally, under bounded domination, the established
cutoff-rate identity gives
\[
 C_\kappa
 =
 D_{1/(1-\kappa)}(\psi\Vert\omega)
 \qquad (0<\kappa<1)
\]
\cite{Junge:2025converse}.
By the definition and integral representation of
$\Phi_1$,
\[
 D_{1/(1-\kappa)}(\psi\Vert\omega)
 =
 \frac{\Phi_1(\kappa)}{\kappa}
 =
 \frac1\kappa\int_0^\kappa
 R_1\!\left(\frac1{1-u}\right)\,\dd u,
\]
which proves~(iii).
\end{proof}

\subsection{Work reliability and its curvature}
\label{subsec:work-reliability}

Work reliability quantifies how rapidly the extraction error can
vanish when a prescribed amount of work per copy is required.
Let $\rho$ be a density matrix on a finite-dimensional system with
Hamiltonian $H$, and let
\begin{equation}
 \gamma=
 \frac{\mathrm e^{-\beta_{\rm th}H}}
      {\Tr\mathrm e^{-\beta_{\rm th}H}},
 \qquad \beta_{\rm th}>0,
 \label{eq:application-gibbs-state}
\end{equation}
be its Gibbs state.

We use the thermal-operation and thermal-battery model of
\cite{Watanabe:2026reliability}. For dimensionless work $w>0$,
the battery $X_w$ has Hamiltonian and Gibbs state
\[
 H_{X_w}
 =\beta_{\rm th}^{-1}\log(\mathrm e^w-1)
   |1\rangle\!\langle1|,
 \qquad
 \mu_w=(1-\mathrm e^{-w})|0\rangle\!\langle0|
       +\mathrm e^{-w}|1\rangle\!\langle1|.
\]
Charging it from $\mu_w$ to $|1\rangle\!\langle1|$ increases
its free energy by
$\beta_{\rm th}^{-1}D(|1\rangle\!\langle1|\Vert\mu_w)
=w/\beta_{\rm th}$, which is the assigned work.
The optimal extraction error is
\begin{equation}
 \mathcal E_{\rm TO}(\rho;w)
 =1-\sup_{\Lambda}
   \langle1|\Lambda(\rho\otimes\mu_w)|1\rangle,
 \label{eq:thermal-battery-error}
\end{equation}
where $\Lambda$ ranges over thermal operations from the system
and battery to the battery, including the limits allowed in
\cite{Watanabe:2026reliability}. We use squared fidelity, so the
matrix element above is the fidelity with the pure target.

For $n$ copies, with additive system Hamiltonian, the reliability
exponent at dimensionless work rate $r>0$ is
\begin{equation}
 B(r)=
 \sup_{\substack{w_n>0\\
                 \liminf_{n\to\infty}w_n/n\geq r}}
 \liminf_{n\to\infty}
 \left[-\frac1n
 \log\mathcal E_{\rm TO}(\rho^{\otimes n};w_n)\right],
 \label{eq:work-reliability-definition}
\end{equation}
with $-\log0=+\infty$.
Thus $B(r)$ is the optimal exponential decay rate of the
extraction error when the asymptotic work per copy is at least
$r/\beta_{\rm th}$.

The operational reliability theorem
\cite[Theorem~2]{Watanabe:2026reliability} gives
\begin{equation}
 B(r)=\sup_{0<\alpha<1}\frac{1-\alpha}{\alpha}
       \bigl[D_\alpha(\rho\Vert\gamma)-r\bigr].
 \label{eq:work-reliability-renyi}
\end{equation}
The optimization includes orders below $1/2$, where
this divergence need not satisfy data processing under arbitrary
quantum channels. Our integral representation and response identity
apply throughout $0<\alpha<1$. They identify the optimizing order
and determine the local dependence of reliability on the work rate.

\begin{corollary}[Optimizing order and curvature of work reliability]
\label{prop:work-reliability-curvature}
Suppose that $\rho>0$ and $\rho\neq\gamma$, and write
\[
 R_1(t)=D(\rho_t^s\Vert\gamma),
 \qquad
 L_t=\log\rho_t^s-\log\gamma.
\]
Then the following statements hold.

\begin{enumerate}
\item
For every $0<r<D(\rho\Vert\gamma)$,
\begin{equation}
 B(r)=\max_{0<\alpha<1}
       \int_\alpha^1\frac{R_1(t)-r}{t^2}\,\dd t.
 \label{eq:work-reliability-area}
\end{equation}
The maximizing order $\alpha_r\in(0,1)$ is unique and is
characterized by
\begin{equation}
 R_1(\alpha_r)=r.
 \label{eq:work-unique-optimizer}
\end{equation}

\item
The functions $r\mapsto\alpha_r$ and $r\mapsto B(r)$ are smooth on
$(0,D(\rho\Vert\gamma))$. Moreover,
\begin{align}
 \frac{\dd\alpha_r}{\dd r}
 &=\frac{\alpha_r}
 {\mathcal V_{\alpha_r}^{\rho_{\alpha_r}^s}(L_{\alpha_r})}>0,
 \label{eq:work-order-response}\\
 B'(r)&=1-\frac1{\alpha_r}<0,
 \label{eq:work-reliability-slope}\\
 B''(r)&=
 \frac{1}
 {\alpha_r\,
  \mathcal V_{\alpha_r}^{\rho_{\alpha_r}^s}(L_{\alpha_r})}>0.
 \label{eq:work-reliability-curvature}
\end{align}
Thus the optimizing order increases with the required work rate,
while the reliability function is strictly decreasing and strictly
convex.

\item
Define the half-order escort and its associated rate by
\begin{equation}
 \eta_*=
 \frac{(\gamma^{1/2}\rho\gamma^{1/2})^{1/2}}
      {\operatorname{Tr}(\gamma^{1/2}\rho\gamma^{1/2})^{1/2}},
 \qquad
 r_*=D(\eta_*\Vert\gamma),
 \qquad
 L_*=\log\eta_*-\log\gamma.
 \label{eq:work-half-order-escort}
\end{equation}
Then $0<r_*<D(\rho\Vert\gamma)$, $\alpha_{r_*}=1/2$, and
\begin{equation}
 B''(r_*)=
 \frac{2}
 {\operatorname{Var}_{\eta_*}(L_*)
  +\frac{\beta_{\rm th}^2}{4}F_Q(\eta_*,H)}.
 \label{eq:work-curvature-fisher}
\end{equation}
Here $F_Q$ is the symmetric-logarithmic-derivative quantum Fisher
information with the convention
\eqref{eq:sld-fisher-convention}.
Furthermore,
\[
 \alpha_r<\frac12
 \quad\Longleftrightarrow\quad
 r<r_*,
\]
and $B$ is smooth across $r=r_*$.
\end{enumerate}
\end{corollary}

\begin{proof}
The integral representation
\eqref{eq:main-fixed-ray-integral}, applied on the sandwiched ray,
gives
\[
 \frac{1-\alpha}{\alpha}D_\alpha(\rho\Vert\gamma)
 =\int_\alpha^1\frac{R_1(t)}{t^2}\,\dd t.
\]
Since
$(1-\alpha)/\alpha=\int_\alpha^1 t^{-2}\,\dd t$,
substitution into \eqref{eq:work-reliability-renyi} yields the area
representation with a supremum in place of the maximum.

By Proposition~\ref{prop:escort-susceptibility},
\begin{equation}
 R_1'(\alpha)=
 \frac1\alpha\,
 \mathcal V_\alpha^{\rho_\alpha^s}(L_\alpha)>0.
 \label{eq:work-escort-response}
\end{equation}
The endpoint values are
\[
 R_1(1)=D(\rho\Vert\gamma),
 \qquad
 \lim_{\alpha\downarrow0}R_1(\alpha)=0.
\]
To establish the second limit, set
$K=\gamma^{-1/2}\rho\gamma^{-1/2}$.
Faithfulness gives constants $0<m\leq M$ such that
$m\one\leq K\leq M\one$.
For $0<\alpha\leq1$, operator monotonicity of the power
$x\mapsto x^\alpha$ therefore implies
\[
 m^\alpha\gamma
 \leq
 \bigl(\gamma^{1/(2\alpha)}
       K\gamma^{1/(2\alpha)}\bigr)^\alpha
 \leq
 M^\alpha\gamma.
\]
After normalization, this proves
$\rho_\alpha^s\to\gamma$ as $\alpha\downarrow0$.
Continuity of relative entropy with faithful reference gives
$R_1(\alpha)\to0$.

Consequently, for every $r\in(0,D(\rho\Vert\gamma))$, there is
exactly one $\alpha_r\in(0,1)$ satisfying $R_1(\alpha_r)=r$.
The derivative of the objective in
\eqref{eq:work-reliability-area} is
\[
 \frac{\dd}{\dd\alpha}
 \int_\alpha^1\frac{R_1(t)-r}{t^2}\,\dd t
 =\frac{r-R_1(\alpha)}{\alpha^2}.
\]
It is positive for $\alpha<\alpha_r$ and negative for
$\alpha>\alpha_r$. Thus $\alpha_r$ is the unique global maximizer,
proving part~(i).

The escort profile is smooth for $\alpha>0$.
Equation~\eqref{eq:work-escort-response} and the inverse function
theorem show that $\alpha_r$ is smooth, with
\[
 \frac{\dd\alpha_r}{\dd r}
 =\frac{1}{R_1'(\alpha_r)}
 =\frac{\alpha_r}
 {\mathcal V_{\alpha_r}^{\rho_{\alpha_r}^s}(L_{\alpha_r})}.
\]
Differentiating the optimized integral, the contribution from its
lower endpoint vanishes because $R_1(\alpha_r)=r$. Hence
\[
 B'(r)
 =-\int_{\alpha_r}^1\frac{\dd t}{t^2}
 =1-\frac1{\alpha_r}.
\]
Differentiating once more gives
\[
 B''(r)
 =\frac{1}{\alpha_r^2}\frac{\dd\alpha_r}{\dd r}
 =\frac{1}
 {\alpha_r\,
  \mathcal V_{\alpha_r}^{\rho_{\alpha_r}^s}(L_{\alpha_r})},
\]
which proves part~(ii).

For part~(iii), observe that $\eta_*=\rho_{1/2}^s$.
Strict monotonicity of $R_1$ gives the assertions about $r_*$
and the position of the optimizer.
The spectral identity used in
Corollary~\ref{cor:order-slopes-fisher} gives
\[
 \mathcal V_{1/2}^{\eta_*}(L_*)
 =
 \operatorname{Var}_{\eta_*}(L_*)
 +\frac14F_Q(\eta_*,\log\gamma).
\]
Since
$\log\gamma=-\beta_{\rm th}H-(\log Z)\one$,
where $Z=\operatorname{Tr}\mathrm e^{-\beta_{\rm th}H}$,
\[
 F_Q(\eta_*,\log\gamma)
 =\beta_{\rm th}^2F_Q(\eta_*,H).
\]
Substituting $\alpha_{r_*}=1/2$ into
\eqref{eq:work-reliability-curvature} proves
\eqref{eq:work-curvature-fisher}.
\end{proof}

\section{Proof of the integral representation}
\label{sec:proof-alpha-z}

We now prove the integral representation. First we derive an entropy identity in a compatible weighted Haagerup
scale, using interpolation and an exact difference quotient.
Endpoint continuity then permits integration. Section~\ref{subsec:fixed-ray-proof} applies this identity to
the fixed-ray escorts, proving
Theorem~\ref{thm:fixed-ray-integral} and
Corollary~\ref{cor:raywise-monotonicity}.
Section~\ref{subsec:lower-order-proof} treats nonnested supports
by compression and normalization, proving
Proposition~\ref{prop:general-fixed-ray-boundary}.

\subsection{A compatible weighted Haagerup scale}
\label{subsec:compatible-scale}

We begin with a positive element at one exponent and construct the
corresponding elements at all smaller exponents by multiplication with
powers of the reference density. The following lemma provides the three
properties needed later: compatibility of this construction,
log-convexity of its norms, and continuity at the initial exponent.

All identities between closed Haagerup products below use the standard
associativity and generalized H\"older theorems for measurable operators;
see \cite{Terp:1981spaces,Kato:2023spaces}.

\begin{lemma}[Compatible weighted scale]
\label{lem:compatible-weighted-scale}
Let $\omega$ be a faithful normal state, let $r>0$, and let
$0\ne Y_r\in L^r(M)_+$.  For $0<p\leq r$, set
\begin{equation}
  d(p,r)\defeq \frac12\left(\frac1p-\frac1r\right)
  \label{eq:weighted-exponent}
\end{equation}
and put
\[
  T_{p,r}\defeq
  \overline{Y_r^{1/2}h_\omega^{d(p,r)}}\in L^{2p}(M).
\]
Define the positive closed product
\begin{equation}
  Y_p
  \defeq T_{p,r}^*T_{p,r}
  =\overline{h_\omega^{d(p,r)}Y_rh_\omega^{d(p,r)}}
  \in L^p(M)_+.
  \label{eq:weighted-path}
\end{equation}
Then $0\ne Y_p\in L^p(M)_+$ and:
\begin{enumerate}[label=\textup{(\roman*)}]
\item for $0<p\leq q\leq r$,
\begin{equation}
  \overline{h_\omega^{d(p,q)}Y_qh_\omega^{d(p,q)}}=Y_p,
  \qquad
  \|Y_p\|_p\leq\|Y_q\|_q;
  \label{eq:weighted-compatibility}
\end{equation}
\item if
\begin{equation*}
  \frac1{p_\theta}
  =\frac{1-\theta}{p_0}+\frac{\theta}{p_1},
  \qquad 0<p_0<p_1\leq r,\qquad 0<\theta<1,
\end{equation*}
then
\begin{equation}
  \|Y_{p_\theta}\|_{p_\theta}
  \leq\|Y_{p_0}\|_{p_0}^{1-\theta}
       \|Y_{p_1}\|_{p_1}^{\theta};
  \label{eq:weighted-log-convexity}
\end{equation}
\item the anchoring endpoint is continuous:
\begin{equation}
  \lim_{p\uparrow r}\|Y_p\|_p=\|Y_r\|_r.
  \label{eq:weighted-endpoint-continuity}
\end{equation}
Consequently,
\begin{equation}
  u\longmapsto B(u)\defeq \log\|Y_{1/u}\|_{1/u}
  \label{eq:weighted-convex-B}
\end{equation}
is finite and convex for $u>1/r$ and has a finite continuous extension to
$u=1/r$.
\end{enumerate}
\end{lemma}

\begin{proof}
We first verify that the products are well defined and nonzero.
The factorization in \eqref{eq:weighted-path} follows from closed
multiplication and H\"older,
because
\begin{equation}
  \frac1{2p}=\frac1{2r}+d(p,r).
\end{equation}
To verify nonvanishing, the case $p=r$ is immediate, so suppose
$p<r$ and put
\[
  e_n=\one_{[1/n,n]}(h_\omega).
\]
Since $\omega$ is faithful, $e_n\uparrow\one$. If $T_{p,r}=0$, then
right multiplication by the bounded operator
$h_\omega^{-d(p,r)}e_n$ gives
\[
  Y_r^{1/2}e_n=0 .
\]
Let $k_r=\one_{\{0\}}(Y_r)$ be the kernel projection of $Y_r$.
The preceding equality implies $e_n\leq k_r$ for every $n$. Hence
\[
  \one=\bigvee_n e_n\leq k_r,
\]
so $k_r=\one$ and therefore $Y_r=0$. This contradicts the assumption
that $Y_r\neq0$. Thus $T_{p,r}\neq0$, and the factorization in
\eqref{eq:weighted-path} implies that $Y_p\neq0$.

To establish compatibility and interpolation, we use a common analytic
core for the different $L^p$ spaces. Denote the modular automorphism group
of $\omega$ by $\sigma^\omega$, and let
$\mathcal A_\omega\subset M$ be its subalgebra of entire analytic
elements. For $a\in\mathcal A_\omega$, define the closed product
\begin{equation}
  i_p^\omega(a)
  \defeq \overline{h_\omega^{1/(2p)}a h_\omega^{1/(2p)}}.
  \label{eq:compatible-core-embedding}
\end{equation}
By \cite[Lemma~1.1]{Junge:2003inequalities}, the image
$i_p^\omega(\mathcal A_\omega)$ is dense in $L^p(M)$ for every $0<p<\infty$.
For $x=i_q^\omega(a)$ set
\begin{equation}
  J_{p,q}^\omega\bigl(i_q^\omega(a)\bigr)
  \defeq i_p^\omega(a),
  \qquad 0<p\leq q.
  \label{eq:downward-map}
\end{equation}
Since
\begin{equation}
  i_p^\omega(a)
  =\overline{h_\omega^{d(p,q)}
             i_q^\omega(a)
             h_\omega^{d(p,q)}},
\end{equation}
generalized H\"older multiplication gives, with $x=i_q^\omega(a)$ and the evident convention when $d(p,q)=0$,
\begin{align}
  \|J_{p,q}^\omega(x)\|_p
  &\leq
  \|h_\omega^{d(p,q)}\|_{1/d(p,q)}^2\|x\|_q\notag\\
  &=\tr(h_\omega)^{2d(p,q)}\|x\|_q
  =\|x\|_q.
  \label{eq:downward-contraction}
\end{align}
This contraction on the common core is also recorded in
\cite[Proposition~2.6]{Gu:2019interpolation}. For arbitrary $x\in L^q(M)$,
generalized H\"older multiplication shows that the closed product
\begin{equation}
  \mathsf M_{p,q}(x)
  \defeq\overline{h_\omega^{d(p,q)}x h_\omega^{d(p,q)}}
  \label{eq:downward-closed-product}
\end{equation}
belongs to $L^p(M)$ and satisfies
$\|\mathsf M_{p,q}(x)\|_p\leq\|x\|_q$.  Since it agrees with
$J_{p,q}^\omega$ on the dense core, it is precisely the unique continuous
extension of that map.  We continue to denote the extension by
$J_{p,q}^\omega$.

The maps compose as expected on the common core. For
\(a\in\mathcal A_\omega\), their defining relation gives
\[
  \bigl(J^\omega_{p,q} J^\omega_{q,r}\bigr)
  \bigl(i_r^\omega(a)\bigr)
  =
  J^\omega_{p,q}\bigl(i_q^\omega(a)\bigr)
  =
  i_p^\omega(a)
  =
  J^\omega_{p,r}\bigl(i_r^\omega(a)\bigr).
\]
Since \(i_r^\omega(\mathcal A_\omega)\) is dense in \(L^r(M)\) and both
\(J^\omega_{p,q} J^\omega_{q,r}\) and \(J^\omega_{p,r}\) are
continuous maps from \(L^r(M)\) to \(L^p(M)\), the identity extends to
all of \(L^r(M)\):
\begin{equation}
  J^\omega_{p,q} J^\omega_{q,r}
  =
  J^\omega_{p,r}.
  \label{eq:downward-composition}
\end{equation}
By \eqref{eq:weighted-path} and the closed-product realization of
\(J^\omega_{s,r}\), we have
\[
  Y_s=J^\omega_{s,r}(Y_r),
  \qquad 0<s\le r.
\]
Consequently, for \(0<p\le q\le r\), the composition identity gives
\[
  J^\omega_{p,q}(Y_q)
  =
  J^\omega_{p,q}J^\omega_{q,r}(Y_r)
  =
  J^\omega_{p,r}(Y_r)
  =
  Y_p.
\]
In terms of closed products, this is
\[
  h_\omega^{d(p,q)}Y_qh_\omega^{d(p,q)}=Y_p.
\]
Finally, the contractivity of \(J^\omega_{p,q}\) yields
\[
  \|Y_p\|_p
  =
  \|J^\omega_{p,q}(Y_q)\|_p
  \leq
  \|Y_q\|_q.
\]
This proves 
\eqref{eq:weighted-compatibility}.  

For the interpolation assertion, we use
\cite[Theorem~4.1 and Corollary~4.3]{Gu:2019interpolation}. Their all-$p$
interpolation inequality is
\begin{equation}
  \|i_{p_\theta}^\omega(a)\|_{p_\theta}
  \leq
  \|i_{p_0}^\omega(a)\|_{p_0}^{1-\theta}
  \|i_{p_1}^\omega(a)\|_{p_1}^{\theta},
  \label{eq:core-all-p-interpolation}
\end{equation}
proved by quasi-Banach complex interpolation in the crossed product.

To pass from the common core to the given anchor, choose
$a_n\in\mathcal A_\omega$ with
$i_r^\omega(a_n)\to Y_r$ in $L^r(M)$. Contractivity then gives
\begin{equation}
  i_p^\omega(a_n)
  =J_{p,r}^\omega(i_r^\omega(a_n))
  \longrightarrow Y_p
  \quad\text{in }L^p(M).
  \label{eq:core-approximation-all-p}
\end{equation}
The passage to the limit in the norms is immediate for $p\geq1$.
For $0<p<1$, the standard Haagerup $L^p$ quasi-norm satisfies the
$p$-triangle inequality
\begin{equation}
  \|u+v\|_p^p
  \leq
  \|u\|_p^p+\|v\|_p^p,
  \qquad u,v\in L^p(M),
  \label{eq:haagerup-p-triangle}
\end{equation}
see \cite[Theorem~4.9(iii)]{Fack:1986measurable}.  Applying
\eqref{eq:haagerup-p-triangle} to
$x_n=x+(x_n-x)$ and then to $x=x_n+(x-x_n)$ gives
\begin{equation}
  \left|\|x_n\|_p^p-\|x\|_p^p\right|
  \leq
  \|x_n-x\|_p^p.
  \label{eq:p-power-continuity}
\end{equation}
Consequently, $x_n\to x$ in $L^p(M)$ implies
$\|x_n\|_p^p\to\|x\|_p^p$, and hence
$\|x_n\|_p\to\|x\|_p$.
Passing to the limit in \eqref{eq:core-all-p-interpolation} therefore
proves \eqref{eq:weighted-log-convexity} for the arbitrary anchor $Y_r$.

It remains to prove continuity at the anchor. We first prove it on the
common core and then use the same approximation to obtain a uniform
error estimate. Fix $a\in\mathcal A_\omega$ and choose $s>r$.
Contractivity gives
\begin{equation}
  \|i_p^\omega(a)\|_p\leq\|i_r^\omega(a)\|_r,
  \qquad p<r.
  \label{eq:core-endpoint-upper}
\end{equation}
Interpolate to $r$ between $p$ and $s$ to obtain
\begin{equation}
  \|i_r^\omega(a)\|_r
  \leq\|i_p^\omega(a)\|_p^{1-\theta(p)}
       \|i_s^\omega(a)\|_s^{\theta(p)},
  \qquad \theta(p)\longrightarrow0
  \quad(p\uparrow r),
  \label{eq:core-endpoint-lower}
\end{equation}
with $\theta(p)$ defined to satisfy $r^{-1}=(1-\theta(p))p^{-1}+\theta(p)s^{-1}$. 
Combining \eqref{eq:core-endpoint-upper} and \eqref{eq:core-endpoint-lower},
\begin{equation}
  \lim_{p\uparrow r}\|i_p^\omega(a)\|_p
  =\|i_r^\omega(a)\|_r.
  \label{eq:core-endpoint-continuity}
\end{equation}

For the approximating sequence chosen above, set
$\delta_n=\|Y_r-i_r^\omega(a_n)\|_r\to0$.
Contractivity controls the error simultaneously at all smaller exponents:
\begin{equation}
  \|Y_p-i_p^\omega(a_n)\|_p\leq\delta_n
  \quad\text{uniformly for }p\leq r.
  \label{eq:endpoint-uniform-error}
\end{equation}
If $r>1$, take $p$ close enough to $r$ that $p\geq1$ and use the reverse
triangle inequality together with
\eqref{eq:core-endpoint-continuity}; this gives
\begin{equation}
  \limsup_{p\uparrow r}
  \left|\|Y_p\|_p-\|Y_r\|_r\right|
  \leq2\delta_n.
\end{equation}
Letting $n\to\infty$ proves the claim.

If $0<r\leq1$, take $p\in[r/2,r]$ and $0<\delta_n\leq1$.  By
\eqref{eq:p-power-continuity} and
\eqref{eq:endpoint-uniform-error},
\begin{equation}
  \left|\|Y_p\|_p^p-
  \|i_p^\omega(a_n)\|_p^p\right|
  \leq\delta_n^p\leq\delta_n^{r/2}.
\end{equation}
The analogous error at $p=r$ is at most $\delta_n^r$.
Combining these two scalar error bounds with
\eqref{eq:core-endpoint-continuity} gives
\begin{equation}
  \limsup_{p\uparrow r}
  \left|\|Y_p\|_p^p-\|Y_r\|_r^r\right|
  \leq\delta_n^{r/2}+\delta_n^r.
\end{equation}
Letting $n\to\infty$ proves convergence of the $p$th powers.  Since
$p\to r>0$ and $Y_r\ne0$, taking the varying root gives
$\|Y_p\|_p\to\|Y_r\|_r$, proving endpoint continuity also in the
quasi-Banach case.  Finally, \eqref{eq:weighted-log-convexity} is exactly convexity of
$B$ in the reciprocal variable, and
\eqref{eq:weighted-endpoint-continuity} supplies its endpoint value.
\end{proof}

\begin{remark}[The quasi-Banach input]
\label{rem:quasi-banach-input}
The use of quasi-Banach interpolation is confined to
\eqref{eq:core-all-p-interpolation}, taken from \cite{Gu:2019interpolation}.
When the entire $p$ interval lies in $[1,\infty)$, ordinary Banach
interpolation \cite{Kosaki:1984interpolation,Terp:1981spaces} supplies the same inequality.
For the fixed-ray choices below, this covers the lower branch when
$z\geq1$ and the upper branch when $z\geq\alpha$.
\end{remark}

Endpoint continuity is needed because the relative entropy at the anchor
may be infinite. The following convex-analysis fact allows us to
integrate the derivative nevertheless.

\begin{lemma}[Convex endpoint fundamental theorem]
\label{lem:convex-endpoint-ftc}
Let $f$ be finite and convex on $(a,b)$ and suppose that it has a finite
continuous extension to $a$.  Then its extension is absolutely continuous
on every compact interval $[a,b']\subset[a,b)$, and
\begin{equation}
  f(b')-f(a)=\int_a^{b'}f'_+(u)\,\dd u
             =\int_a^{b'}f'_-(u)\,\dd u.
  \label{eq:convex-endpoint-ftc}
\end{equation}
The two one-sided derivatives agree at all but at most countably many points.
\end{lemma}

\begin{proof}
The interior integral identity and the properties of the one-sided
derivatives follow from
\cite[Theorem~24.1 and Corollary~24.2.1]{Rockafellar:1970convex}.
Apply the identity on $[a+\varepsilon,b']$ and let
$\varepsilon\downarrow0$. Since $f'_\pm$ are nondecreasing and
bounded above on $(a,b')$, continuity of $f$ at $a$ implies
$f'_\pm\in L^1(a,b')$ and gives
\eqref{eq:convex-endpoint-ftc}. The integral representation then
implies absolute continuity.
\end{proof}

We now identify the derivative of the logarithmic norm. The main step is
an exact expression for a difference quotient in terms of a lower-order
$\alpha$--$z$ divergence. Taking its order to one gives Araki relative
entropy. Convexity in the reciprocal exponent then gives both
monotonicity of the entropy profile and its integral representation.
The statement applies to any anchor $0\ne Y_r\in L^r(M)_+$, and does not 
yet use the fixed-ray factors.

\begin{proposition}[Anchored entropy identity]
\label{prop:anchored-entropy-identity}
Let $\omega$ be a faithful normal state on $M$, let $r>0$, and let
$0\ne Y_r\in L^r(M)_+$.  For $0<p\leq r$, let
\begin{equation}
  Y_p
  \defeq
  \overline{
    h_\omega^{\frac12(\frac1p-\frac1r)}
    Y_r
    h_\omega^{\frac12(\frac1p-\frac1r)}}
  \in L^p(M)_+
  \label{eq:anchored-entropy-path}
\end{equation}
be the compatible path of
Lemma~\ref{lem:compatible-weighted-scale}, and set
\begin{equation}
  Z(p)\defeq\tr(Y_p^p),
  \qquad
  G(p)\defeq\frac1p\log Z(p)=\log\|Y_p\|_p,
  \qquad
  h_{\eta_p}\defeq\frac{Y_p^p}{Z(p)}.
  \label{eq:anchored-entropy-data}
\end{equation}
Then $Z(p)\in(0,\infty)$ and $\eta_p$ is a normal state.  Moreover:
\begin{enumerate}[label=\textup{(\roman*)}]
\item For every $0<p'<p\leq r$,
\begin{align}
  Q_{p'/p,p'}(\eta_p\|\omega)
  &=Z(p')Z(p)^{-p'/p},
  \label{eq:anchored-exact-quotient}\\
  \frac{G(p')-G(p)}{p'-p}
  &=\frac{D_{p'/p,p'}(\eta_p\Vert\omega)}{pp'}.
  \label{eq:anchored-exact-difference-quotient}
\end{align}

\item For every $0<p\leq r$, the left derivative exists in the
extended-real sense and satisfies
\begin{equation}
  G'_-(p)
  =\frac{D(\eta_p\Vert\omega)}{p^2}.
  \label{eq:anchored-entropy-left-derivative}
\end{equation}
For $p<r$ both sides are finite.  At the anchoring endpoint $p=r$
they may equal $+\infty$.

\item The map
\begin{equation}
  p\longmapsto D(\eta_p\Vert\omega)
  \label{eq:anchored-entropy-profile}
\end{equation}
is nondecreasing and Borel measurable as a
$[0,+\infty]$-valued function on $(0,r]$.

\item The function $G$ is absolutely continuous on every compact interval
$[p_0,p_1]\subset(0,r]$, and
\begin{equation}
  G(p_1)-G(p_0)
  =
  \int_{p_0}^{p_1}
  \frac{D(\eta_p\Vert\omega)}{p^2}\,\dd p,
  \qquad 0<p_0<p_1\leq r.
  \label{eq:anchored-entropy-integral}
\end{equation}
In particular,
\begin{equation}
  \log\frac{\|Y_r\|_r}{\|Y_{p_0}\|_{p_0}}
  =
  \int_{p_0}^{r}
  \frac{D(\eta_p\Vert\omega)}{p^2}\,\dd p.
  \label{eq:anchored-entropy-integral-to-anchor}
\end{equation}
These integrals are finite.  If
$D(\eta_r\Vert\omega)=+\infty$, its value at the single endpoint $r$
does not affect the Lebesgue integral.
\end{enumerate}
\end{proposition}

\begin{proof}
Lemma~\ref{lem:compatible-weighted-scale} gives
$0\ne Y_p\in L^p(M)_+$.  Hence
$Z(p)\in(0,\infty)$ and $Z(p)^{-1}Y_p^p$ is the Haagerup density of a
normal state.

To prove the quotient identity, fix $0<p'<p\leq r$ and put
\[
  a\defeq\frac{p'}p\in(0,1).
\]
Functional calculus gives
\begin{align}
  \overline{
  h_\omega^{\frac{1-a}{2p'}}
  h_{\eta_p}^{a/p'}
  h_\omega^{\frac{1-a}{2p'}}}
  =
  Z(p)^{-1/p}
  \overline{
    h_\omega^{\frac12(\frac1{p'}-\frac1p)}
    Y_p
    h_\omega^{\frac12(\frac1{p'}-\frac1p)}}.
\end{align}
Compatibility of the anchored path identifies the closed product on the
right with $Y_{p'}$.  Taking the $p'$-th moment therefore gives
\eqref{eq:anchored-exact-quotient}.  Since $\eta_p$ and $\omega$ are
states,
\begin{align}
  D_{p'/p,p'}(\eta_p\Vert\omega)
  &=
  \frac{1}{p'/p-1}
  \log\!\left(Z(p')Z(p)^{-p'/p}\right)\\
  &=
  pp'\frac{G(p')-G(p)}{p'-p},
\end{align}
which proves \eqref{eq:anchored-exact-difference-quotient}.

To identify the left derivative, let $p'\uparrow p$ in this identity.
Both parameters of the divergence vary in the limit, so we compare its
second parameter with two fixed values. For $p'$ sufficiently close to
$p$, so that $p/2<p'<p$, monotonicity in the second parameter at lower order
\cite[Theorem~1 (x)]{Kato:2023divergence} gives
\begin{equation}
  D_{p'/p,p/2}(\eta_p\Vert\omega)
  \leq
  D_{p'/p,p'}(\eta_p\Vert\omega)
  \leq
  D_{p'/p,p}(\eta_p\Vert\omega).
  \label{eq:anchored-varying-z-squeeze}
\end{equation}
For the fixed second parameters $p/2$ and $p$, the lower-order endpoint
theorem \cite[Theorem~6.7]{Hiai:2024reversibility} yields
\begin{equation}
  \lim_{p'\uparrow p}D_{p'/p,p/2}(\eta_p\Vert\omega)
  =
  \lim_{p'\uparrow p}D_{p'/p,p}(\eta_p\Vert\omega)
  =
  D(\eta_p\Vert\omega)
\end{equation}
in $[0,+\infty]$.  Squeezing in
\eqref{eq:anchored-exact-difference-quotient} proves
\eqref{eq:anchored-entropy-left-derivative}.

The remaining assertions follow from convexity in the reciprocal
exponent. Define
\begin{equation}
  B(u)\defeq G(1/u)=\log\|Y_{1/u}\|_{1/u},
  \qquad u\geq\frac1r.
\end{equation}
By Lemma~\ref{lem:compatible-weighted-scale}, $B$ is finite and convex on
$(1/r,\infty)$ and continuous at $1/r$.  The one-sided chain rule gives
\begin{equation}
  B'_+(1/p)
  =
  -p^2G'_-(p)
  =
  -D(\eta_p\Vert\omega).
  \label{eq:anchored-B-derivative}
\end{equation}
If $p<r$, then $1/p$ is an interior point of the domain of the finite
convex function $B$, so its one-sided derivative, and hence
$D(\eta_p\Vert\omega)$, is finite.

Since $B'_+$ is nondecreasing in $u$, reversing the reciprocal variable
in \eqref{eq:anchored-B-derivative} shows that
$p\mapsto D(\eta_p\Vert\omega)$ is nondecreasing.  It is therefore Borel
measurable.

Finally, Lemma~\ref{lem:convex-endpoint-ftc} shows that $B$ is absolutely
continuous on every compact interval $[1/r,b]$.  Since $p\mapsto1/p$ is
bi-Lipschitz on compact subintervals of $(0,r]$, $G(p)=B(1/p)$ is
absolutely continuous on $[p_0,p_1]$.  Its ordinary derivative agrees
almost everywhere with the left derivative in
\eqref{eq:anchored-entropy-left-derivative}.  The fundamental theorem of
calculus proves \eqref{eq:anchored-entropy-integral}; taking $p_1=r$
gives \eqref{eq:anchored-entropy-integral-to-anchor}.  Finiteness follows
from the finiteness of the endpoint values of $G$.
\end{proof}

\subsection{Proof of the fixed-ray representation}
\label{subsec:fixed-ray-proof}

To apply Proposition~\ref{prop:anchored-entropy-identity}, we must choose
an anchor whose weighted path agrees with the factors defining the
$\alpha$--$z$ divergence. The anchor is at order one for $\alpha<1$ and
at the final order for $\alpha>1$. Once this identification is made,
the integral formula and the refined derivative identity follow from the
same proposition.

Put $e=s(\omega)$, pass to the corner $eMe$, and, for the duration of the
proof, relabel the corner and the restricted states as $M$, $\psi$, and
$\omega$.  Thus $\omega$ is faithful.  The states constructed in the
corner extend canonically to the original algebra, and the corresponding
relative entropies are unchanged.

Recall that $c=z/\alpha$.  Choose the anchor
\begin{equation}
  (r,Y_r)\defeq
  \begin{cases}
    \bigl(c,h_\psi^{1/c}\bigr),
      &0<\alpha<1,\\[0.5ex]
    \bigl(z,x_{\alpha,z}(\psi\|\omega)\bigr),
      &\alpha>1.
  \end{cases}
  \label{eq:fixed-ray-anchor}
\end{equation}
For $\alpha<1$, support inclusion makes $h_\psi^{1/c}$ an element of the
corner in which we work. For $\alpha>1$, the finiteness hypothesis
ensures that the factor $x_{\alpha,z}(\psi\|\omega)$ exists and belongs
to $L^z(M)_+$.

Let $(Y_p)_{0<p\leq r}$ be the associated anchored path, and use the
notation $Z(p)$, $G(p)$, and $\eta_p$ of Proposition~\ref{prop:anchored-entropy-identity}.  In the lower branch, for $\beta\in[\alpha,1]$,
\[
    Y_{c\beta} = J_{c\beta,c}^{\omega}(h_\psi^{1/c}) = h_\omega^{\frac{1-\beta}{2c\beta}} h_\psi^{1/c}h_\omega^{\frac{1-\beta}{2c\beta}} = x_\beta^{(c)}(\psi\|\omega).
\]
In the upper branch, 
the defining factorization \eqref{eq:upper-fixed-ray-factor} for the anchor $Y_z=x_{\alpha,z}(\psi\|\omega)$ is
\begin{equation}
  h_\psi^{1/c}
  =
  \overline{
  h_\omega^{\frac1{2c}-\frac1{2z}}
  Y_z
  h_\omega^{\frac1{2c}-\frac1{2z}}}
  =
  J_{c,z}^\omega(Y_z).
  \label{eq:upper-fixed-ray-anchor-factorization}
\end{equation}
For $\beta\in[1,\alpha]$, we have
\begin{equation}
    h_\psi^{1/c}
    =
    J_{c,z}^\omega(Y_z)
    = 
    J_{c,c\beta}^\omega J_{c\beta,z}^\omega(Y_z)
    = 
    J_{c,c\beta}^\omega(Y_{c\beta})
    =
    h_\omega^{\frac{\beta-1}{2c\beta}}Y_{c\beta}h_\omega^{\frac{\beta-1}{2c\beta}}.
\end{equation}
By uniqueness, $Y_{c\beta}=x_\beta^{(c)}$ also in the upper branch.

Taking powers, traces and normalized densities then gives
\begin{equation}
  Z(c\beta)
  =
  Q_{\beta,c\beta}(\psi\|\omega),
  \qquad
  \eta_{c\beta}
  =
  \psi_\beta^{(c)},
  \qquad \beta\in I_\alpha.
  \label{eq:anchored-fixed-ray-identification}
\end{equation}
This proves the existence and finite-nonzero-moment assertions of
Theorem~\ref{thm:fixed-ray-integral}, and identifies the normalized powers
with the required escorts. We can therefore apply the anchored entropy
identity to obtain \eqref{eq:main-fixed-ray-integral}.

At the two distinguished endpoints,
\begin{equation}
  Z(c)=\tr(h_\psi)=1,
  \qquad
  Z(z)=Q_{\alpha,z}(\psi\|\omega),
  \qquad
  G(c)=0.
  \label{eq:fixed-ray-anchor-endpoints}
\end{equation}

Applying \eqref{eq:anchored-entropy-integral} between $p=c$ and $p=z$
gives
\begin{align}
  \frac1z\log Q_{\alpha,z}(\psi\|\omega)
  &=
  G(z)-G(c)\notag\\
  &=
  \int_c^z
  \frac{D(\eta_p\Vert\omega)}{p^2}\,\dd p\notag\\
  &=
  \frac1c
  \int_1^\alpha
  \frac{D(\psi_\beta^{(c)}\Vert\omega)}{\beta^2}\,\dd\beta.
  \label{eq:fixed-ray-from-anchored-identity}
\end{align}
Here the integrals are oriented; when $0<\alpha<1$, this identity is
obtained by applying
\eqref{eq:anchored-entropy-integral} on $[z,c]$ and reversing the signs.
Since $z=c\alpha$, multiplication by $z/(\alpha-1)$ yields
\begin{equation}
  D_{\alpha,z}(\psi\Vert\omega)
  =
  \frac{\alpha}{\alpha-1}
  \int_1^\alpha
  \frac{D(\psi_\beta^{(c)}\Vert\omega)}{\beta^2}\,\dd\beta.
\end{equation}
This is \eqref{eq:main-fixed-ray-integral}. Rewriting the oriented
integral with the positive measure $\mu_\alpha$ gives the stated average
and completes the proof of Theorem~\ref{thm:fixed-ray-integral}.

We next prove Corollary~\ref{cor:raywise-monotonicity}. On every interval
of the fixed-ray parameter obtained from a single anchor,
Proposition~\ref{prop:anchored-entropy-identity} shows that
$p\mapsto D(\eta_p\Vert\omega)$ is nondecreasing.  Since $p=c\beta$,
\eqref{eq:anchored-fixed-ray-identification} implies that
$\beta\mapsto R_c(\beta)$ is nondecreasing on that segment.

These comparisons give monotonicity throughout $\mathcal J_c$.
Indeed, if
$0<\beta_1<\beta_2\leq1$, use the lower path anchored at $p=c$.  If
$1\leq\beta_1<\beta_2$ and $\beta_2\in\mathcal J_c$, use the upper path
anchored at $p=c\beta_2$.  If $\beta_1<1<\beta_2$, combine the preceding
two comparisons through $R_c(1)$.  Hence $R_c$ is nondecreasing on
$\mathcal J_c$. As a monotone extended-real-valued function 
on the interval \(\mathcal J_c\), \(R_c\) is Borel measurable.

For the derivative identity \eqref{eq:raywise-refined-identity}, observe
that on any such anchored segment,
\begin{equation}
  F_c(\beta)
  =
  \frac1\beta\log Z(c\beta)
  =
  cG(c\beta).
  \label{eq:fixed-ray-potential-from-anchor}
\end{equation}
Therefore \eqref{eq:anchored-entropy-left-derivative} and the one-sided
chain rule give
\begin{align}
  \beta^2\frac{\dd^-}{\dd\beta}F_c(\beta)
  &=
  \beta^2c^2G'_-(c\beta)\notag\\
  &=
  D(\eta_{c\beta}\Vert\omega)
  =
  R_c(\beta).
\end{align}
If \(\beta>1\) is an interior point of \(\mathcal J_c\), choose
\(\gamma\in\mathcal J_c\) with \(\gamma>\beta\).  Applying
Proposition~\ref{prop:anchored-entropy-identity} to the upper path anchored at \(r=c\gamma\) places
\(p=c\beta\) strictly below the anchor.  Hence
\(R_c(\beta)=D(\eta_{c\beta}\Vert\omega)\) is finite, and the derivative
identity applies there as an interior identity.
Points below one belong to the lower path.  At
$\beta=1$, the lower path gives the asserted endpoint left derivative,
with the extended-real interpretation when
$D(\psi\Vert\omega)=+\infty$.  

If \(M\) is finite-dimensional, then after passing to the support
corner \(h_\omega\) is invertible and
\[
  x_\beta^{(c)}
  =
  h_\omega^{(1-\beta)/(2c\beta)}
  h_\psi^{1/c}
  h_\omega^{(1-\beta)/(2c\beta)}
\]
is a smooth positive-matrix family of constant rank. Smooth
finite-dimensional functional calculus therefore shows that
\(F_c\) and \(R_c\) are smooth on the interior of \(\mathcal J_c\), so the
left derivative there agrees with the ordinary derivative.
This completes the proof of
Corollary~\ref{cor:raywise-monotonicity}.

\subsection{The lower-order support boundary}
\label{subsec:lower-order-proof}

We finish by removing the support-inclusion assumption in the lower-order
range. Compression to the support of the reference state gives the same
factors, up to a scalar normalization. This scalar cancels from the
escort states but remains in the divergence as the boundary term of
Proposition~\ref{prop:general-fixed-ray-boundary}.

Let $e=s(\omega)$ and use the canonical corner identification.  Since
$x\mapsto{exe}$ is a positive contraction on $L^c(M)$,
$y_c\in L^c(eMe)_+$ and $q_c<\infty$.  Faithfulness of the canonical
functional on $L^1(eMe)_+$ also gives $q_c=0$ only if $y_c=0$.
Moreover, contractivity of the compression gives
\[
  q_c
  = \|y_c\|_c^c
  = \bigl\|e h_\psi^{1/c}e\bigr\|_c^c
  \leq \bigl\|h_\psi^{1/c}\bigr\|_c^c
  = \tr(h_\psi)
  = 1.
\]
This proves \eqref{eq:boundary-cost-bound}, namely \(0\leq q_c\leq1\).

To identify the compressed factors, use the fact that $h_\omega$ is
supported on $e$. For $0<\beta<1$, we may insert $e$ on both sides of
$h_\psi^{1/c}$ in the closed product:
\begin{align}
  x_\beta^{(c)}(\psi\|\omega)
  &=\overline{
  h_\omega^{\frac{1-\beta}{2c\beta}}
  y_c
  h_\omega^{\frac{1-\beta}{2c\beta}}}.
  \label{eq:general-boundary-factor-compression}
\end{align}
If $q_c>0$, then
$y_c=q_c^{1/c}h_{\widehat\psi_c}^{1/c}$, and hence
\begin{equation}
  x_\beta^{(c)}(\psi\|\omega)
  =q_c^{1/c}
   x_\beta^{(c)}(\widehat\psi_c\|\omega_e).
  \label{eq:general-boundary-factor-scaling}
\end{equation}
Taking the $c\beta$-th moment gives
\begin{equation}
  Q_{\beta,c\beta}(\psi\|\omega)
  =q_c^\beta
   Q_{\beta,c\beta}(\widehat\psi_c\|\omega_e).
  \label{eq:general-boundary-moment-scaling}
\end{equation}
The factor \(q_c^\beta\) in \eqref{eq:general-boundary-moment-scaling} cancels upon normalization.
Consequently, \eqref{eq:general-boundary-factor-scaling} and \eqref{eq:general-boundary-moment-scaling} prove the escort-identification assertion
of Proposition~\ref{prop:general-fixed-ray-boundary}.  

For \(q_c>0\) at $\beta=\alpha$,
\begin{equation}
  D_{\alpha,c\alpha}(\psi\Vert\omega)
  =D_{\alpha,c\alpha}(\widehat\psi_c\Vert\omega_e)
   -\frac{\alpha}{1-\alpha}\log q_c.
\end{equation}
The state $\widehat\psi_c$ is supported in $e$, so applying
Theorem~\ref{thm:fixed-ray-integral} to the corner proves
\eqref{eq:general-fixed-ray-boundary} when \(q_c>0\).  

If $q_c=0$, then $y_c=0$, and
\eqref{eq:general-boundary-factor-compression} gives
\(Q_{\alpha,c\alpha}(\psi\Vert\omega)=0\), hence
\(D_{\alpha,c\alpha}(\psi\Vert\omega)=+\infty\).
This proves the remaining case and completes the proof of Proposition~\ref{prop:general-fixed-ray-boundary}.

\section{Discussion}
\label{sec:discussion}

The integral representation in Theorem~\ref{thm:fixed-ray-integral} 
expresses $\alpha$--$z$ divergences as
averages of relative entropies along the family of escort states 
$\psi_{\beta}^{(z/\alpha)}$ with $\beta\in I_\alpha$ and $c=z/\alpha$ fixed. This holds for 
the entire range of parameters $\alpha,z$ without
assuming data processing. Above order one, finiteness of the divergence
provides the integrability needed on the intervening orders.
Below order one, compression to the support of the reference state
accounts for an additional term. In applications, the
remaining task is to explicitly construct or estimate the escort states. In the
coherent-state example, their identification makes modular-theoretic
relative-entropy formulas available for the R\'enyi calculation.

For faithful matrices, susceptibility gives a local description of
the escort family as in Proposition~\ref{prop:escort-susceptibility}. 
The response of the escort relative entropy, the
finite-order remainder, and the curvature of the work-reliability
function are determined by the same quadratic fluctuation form. Differences
between its order-one values isolate quantum Fisher and
skew-information terms. For Gibbs states, its integer
values are discrete averages of an imaginary-time correlation
function and converge to the Bogoliubov--Kubo--Mori covariance.
Optimization over a conditioning state leads to a second family of
integral representations, described in 
Proposition~\ref{prop:optimized-conditional-integral}.
The stationary reference state equals the
escort marginal, so ordinary conditional entropy becomes the
integrand. The resulting formula holds for faithful bipartite matrices,
every fixed ray, and all orders above one. It requires neither data
processing nor uniqueness of the optimizer. On the sandwiched ray,
it recovers the conditional-entropy identity of
\cite[Lemma 7]{Cheng:2026transmission}. 
It would be useful to extend these identities to normal
states on general von Neumann algebras by identifying natural
domain conditions in $L^p(M)$.
An algebraic extension of the optimized conditional-entropy formula
would also require control of the minimizing references and a suitable
envelope argument. 

The order-slope identities also admit an interpretation in terms of
modular flow. For faithful density matrices, let $K_\sigma=-\log\sigma$ and
consider the trajectory $\rho_s=e^{-isK_\sigma}\rho e^{isK_\sigma}$.
With the standard normalization of the Bures metric, its squared
speed is \cite{Braunstein:1994statistical,Taddei:2013speed}
\[
 v_B^2=\frac14 F_Q(\rho,K_\sigma).
\]
Combining this relation with \eqref{eq:slope-qfi} gives
\[
 S_{1/2}(\rho\Vert\sigma)-S_1(\rho\Vert\sigma)
 =\frac12 v_B^2.
\]
The difference therefore measures the local sensitivity of $\rho$
to the modular flow generated by the reference state $\sigma$.
This identity raises the question of whether the difference
$S_{1/2}-S_1$ admits a bulk geometric representation in holographic
conformal field theories, where the dual to the Bures metric is known in some cases. 
A possible setting is provided by states in the
vacuum conformal orbit of a two-dimensional CFT, for which a
geometric expression for the Bures metric is known
\cite{Gerbershagen:2024bures}.
Applying that expression would require a continuum version of
the order-one slope identity and identification of the tangent
generated by the reference modular flow within the conformal
family. We leave this question for the future.

A further question concerns the relation between this representation
and decompositions based on quantum hockey-stick divergences \cite{Frenkel:2022integral,Hirche:2023integral,Liu:2025representations}. Both
express nonlinear distinguishability measures through an integral,
but they use different parameters and different integrands. Their
combination may be useful when estimates are available for only one
of these ingredients. The behavior of escort families under channels
and restrictions of the observable algebra also merits study.

\section*{Declarations}
\paragraph*{Funding.}
T.K. is supported by a Simons Foundation fellowship through the Targeted
Grant to Instituto Balseiro. The work of P.R. has been supported by the
Polish National Science Centre through Sonata grant (2022/47/D/ST2/02058).
\paragraph*{Competing interests.}
The authors have no relevant financial or non-financial interests to disclose.
\paragraph*{Data and code availability.}
No experimental datasets were generated. 
\paragraph*{Use of AI tools.}
LLMs were used for refining the technical arguments and
as editorial aids. The authors take full responsibility for the 
contents of this paper. This work was supported in part by a grant of 
access to OpenAI models through the ChatGPT for Academic Researchers 
program.

\appendix

\section{Finite-dimensional support arguments}
\label{app:matrix-support}

We extend the matrix notation of Section~\ref{sec:finite-dimensional} to singular density matrices.  Put
$P=s(\sigma)$.  Positive powers of $\sigma$ are understood by the usual
functional calculus, with $0^u=0$ for $u>0$.  Whenever a negative power
of $\sigma$ occurs, it denotes the inverse power of the faithful
restriction $\sigma_P\defeq\sigma\vert_{P\cH}$ on $P\cH$, extended by
zero on $P^\perp\cH$.  Thus, for $0<\alpha<1$, the expression
$Q_{\alpha,z}(\rho\|\sigma)=\Tr(\sigma^{(1-\alpha)/(2z)}\rho^{\alpha/z}\sigma^{(1-\alpha)/(2z)})^z$ is defined for arbitrary $\rho$ and
$\sigma$, with the convention that
$D_{\alpha,z}(\rho\Vert\sigma)=+\infty$ if
$Q_{\alpha,z}(\rho\|\sigma)=0$.  For $\alpha>1$ the generalized-inverse reading of $\sigma^{(1-\alpha)/2z}$ is used only
when $s(\rho)\le s(\sigma)$. We set
\[
  Q_{\alpha,z}(\rho\|\sigma)
  =D_{\alpha,z}(\rho\Vert\sigma)
  =+\infty
  \qquad\text{if }s(\rho)\not\leq s(\sigma).
\]
If $s(\rho)\leq s(\sigma)$, all the quantities may equivalently be
computed in the corner $P\mathcal B(\cH)P\cong\mathcal B(P\cH)$, where
$\sigma_P$ is faithful.  Matrices and states obtained in this corner
will always be regarded as matrices on $\cH$ by extension by zero on
$P^\perp\cH$.  In particular, under the support inclusion,
$Q_{\beta}^{(c)}(\rho\|\sigma)$ is finite and strictly positive for
every $\beta>0$, and the escort state
$\rho_\beta^{(c)}$ is well defined.

\subsection{Approximation by faithful states}
\label{subsec:matrix-approximation}
We give an elementary proof of the finite-dimensional specialization
of Theorem~\ref{thm:fixed-ray-integral} for supported pairs of states.
Fix $\alpha,z>0$ with $\alpha\ne1$, put $c=z/\alpha$, and assume
$s(\rho)\leq s(\sigma)$.
\begin{proof}
The calculation in Section~\ref{sec:finite-dimensional} proves the result when both
$\rho$ and $\sigma$ are faithful.  For general $\sigma$, put
$P=s(\sigma)$.  Since $s(\rho)\leq P$, we have $\rho=P\rho P$, and all
the moments, divergences, and escort states agree
with the corresponding quantities computed in the corner
$\mathcal B(P\cH)$.  We may therefore replace $\cH$ by $P\cH$ and
$\sigma$ by its faithful restriction
$\sigma_P=\sigma\vert_{P\cH}$.

If $\rho$ is faithful on $P\cH$, the preceding calculation applies
directly.  Otherwise, define
\begin{equation}
  \rho_\varepsilon
  \defeq
  (1-\varepsilon)\rho
  +\varepsilon\frac{P}{\operatorname{rank}P},
  \qquad 0<\varepsilon<1.
  \label{eq:fd-faithful-regularization}
\end{equation}
Then $\rho_\varepsilon$ is a faithful density matrix on $P\cH$, and
hence
\begin{equation}
  D_{\alpha,z}(\rho_\varepsilon\Vert\sigma_P)
  =
  \frac{\alpha}{\alpha-1}
  \int_1^\alpha
  \frac{
    D((\rho_\varepsilon)_\beta^{(c)}\Vert\sigma_P)}
       {\beta^2}\,\dd\beta .
  \label{eq:fd-regularized-ray-integral}
\end{equation}

As $\varepsilon\downarrow0$, we have
$\rho_\varepsilon\to\rho$ in norm.  Continuity of the matrix power map
on the positive cone gives
\[
  \rho_\varepsilon^{1/c}\longrightarrow\rho^{1/c}
\]
in norm.  Consequently, for every fixed $\beta$ in the closed interval
with endpoints $1$ and $\alpha$,
\[
  X_\beta^{(c)}[\rho_\varepsilon]
  \longrightarrow
  X_\beta^{(c)}[\rho]
\]
and
\[
  \bigl(X_\beta^{(c)}[\rho_\varepsilon]\bigr)^{c\beta}
  \longrightarrow
  \bigl(X_\beta^{(c)}[\rho]\bigr)^{c\beta}
\]
in norm.  Taking traces therefore gives
\[
  Q_\beta^{(c)}(\rho_\varepsilon\|\sigma_P)
  \longrightarrow
  Q_\beta^{(c)}(\rho\|\sigma_P).
\]
The limiting moment is strictly positive: in the faithful corner the
powers of $\sigma_P$ are invertible, while $\rho^{1/c}\neq0$.  It
follows that
\[
  (\rho_\varepsilon)_\beta^{(c)}
  \longrightarrow
  \rho_\beta^{(c)}
\]
in norm for every such $\beta$.  In particular, at $\beta=\alpha$,
\[
  Q_{\alpha,z}(\rho_\varepsilon\|\sigma_P)
  \longrightarrow
  Q_{\alpha,z}(\rho\|\sigma_P)>0,
\]
and hence
\begin{equation}
  D_{\alpha,z}(\rho_\varepsilon\Vert\sigma_P)
  \longrightarrow
  D_{\alpha,z}(\rho\Vert\sigma_P).
  \label{eq:fd-regularized-divergence-limit}
\end{equation}

Since $\sigma_P$ is faithful and $P\cH$ is finite dimensional, the map
\[
  \eta\longmapsto D(\eta\Vert\sigma_P)
\]
is continuous on the compact state space of $\mathcal B(P\cH)$.
Therefore
\[
  D((\rho_\varepsilon)_\beta^{(c)}\Vert\sigma_P)
  \longrightarrow
  D(\rho_\beta^{(c)}\Vert\sigma_P)
\]
for every $\beta$.  Moreover, writing
$\lambda_{\min}(\sigma_P)>0$ for the smallest eigenvalue of
$\sigma_P$, every state $\eta$ on $P\cH$ satisfies
\begin{align}
  D(\eta\Vert\sigma_P)
  &=\Tr(\eta\log\eta)-\Tr(\eta\log\sigma_P) \notag\\
  &\leq-\log\lambda_{\min}(\sigma_P).
  \label{eq:fd-relative-entropy-uniform-bound}
\end{align}
Since $\beta^{-2}$ is bounded on the closed interval with endpoints
$1$ and $\alpha$, this gives an integrable bound independent of
$\varepsilon$.  Dominated convergence may therefore be applied to
\eqref{eq:fd-regularized-ray-integral}.  Combining it with
\eqref{eq:fd-regularized-divergence-limit} proves
\eqref{eq:main-fixed-ray-integral} for matrices. The lower-order form
follows by reversing the orientation of
the integral.
\end{proof}
\subsection{Compression and the lower-order boundary}
\label{subsec:fd-support-boundary}

For $0<\alpha<1$, finiteness of {\az} divergence does not necessitate the support inclusion $s(\rho)\leq s(\sigma)$.  The missing
support is then visible as an endpoint jump.  The finite-dimensional
formula in Proposition~\ref{prop:fd-boundary-correction} also clarifies why the general theorem in
Section~\ref{sec:preliminaries} assumes nested supports.

\begin{lemma}[Compression cost]\label{lem:compression}
Let $B\ge0$ be a positive semidefinite matrix on a finite-dimensional Hilbert space
$\mathcal H$, let $P\in\mathcal B(\mathcal H)$ be an orthogonal projection, and let $c>0$.
Then
\begin{equation}\label{eq:compression}
  0\ \le\ \operatorname{Tr}\big[(PBP)^{c}\big]\ \le\ \operatorname{Tr}\big[B^{c}\big],
\end{equation}
with equality on the left if and only if $s(B)\le 1-P$, and on the right if and only if
$s(B)\le P$.
\end{lemma}

\begin{proof}
Order eigenvalues of the compression being computed on $P\mathcal H$ decreasingly,
and extended by zeros. Denote the eigenvalues by $\lambda_i(PBP)$. For $v\in P\mathcal H$ one has $\langle v,PBPv\rangle=\langle v,Bv\rangle$,
so the Courant--Fischer-Weyl principle in the max-min form \cite[Corollary III.1.2]{Bhatia:1997matrix}, applied over subspaces of $P\mathcal H$ rather than of $\mathcal H$, gives
\begin{equation}\label{eq:interlace}
  \lambda_{j}(PBP)\le\lambda_{j}(B)\qquad\text{for every }j.
\end{equation}
Since $t\mapsto t^{c}$ is increasing on $[0,\infty)$ and vanishes at $0$, summing
\eqref{eq:interlace} proves the right-hand inequality in \eqref{eq:compression}; the
left-hand one is trivial. (Only monotonicity is used, so the same argument gives
$\operatorname{Tr}f(PBP)\le\operatorname{Tr}f(B)$ for every increasing $f$ with $f(0)=0$.)

Suppose the right-hand equality holds. Each summand on the left is dominated by the
corresponding summand on the right, so all of them agree, and strict monotonicity of
$t\mapsto t^{c}$ upgrades this to equality in \eqref{eq:interlace} for every $j$. Summing
gives $\operatorname{Tr}[PBP]=\operatorname{Tr}[B]$, that is $\operatorname{Tr}[(1-P)B]=0$,
and $B\ge0$ forces $B^{1/2}(1-P)=0$, i.e.\ $s(B)\le P$. Conversely $s(B)\le P$ gives
$PBP=B$.

Finally, $\operatorname{Tr}[(PBP)^{c}]=0$ holds if and only if $(PBP)^{c}=0$, hence if and
only if $PBP=0$, hence if and only if $B^{1/2}P=0$, i.e.\ $s(B)\le 1-P$.
\end{proof}

\begin{proposition}[Finite-dimensional boundary correction]
\label{prop:fd-boundary-correction}
Let $0<\alpha<1$, $c>0$, and let $P=s(\sigma)$.  Define
\begin{equation}
  q_{1-,c}
  \defeq \Tr\!\left[(P\rho^{1/c}P)^c\right].
  \label{eq:fd-boundary-overlap}
\end{equation}
If $q_{1-,c}=0$, then
$D_{\alpha,c\alpha}(\rho\Vert\sigma)=+\infty$.  If $q_{1-,c}>0$, put
\begin{equation}
  \widehat\rho_c
  \defeq \frac{(P\rho^{1/c}P)^c}{q_{1-,c}}
  \quad\text{on }P\cH.
  \label{eq:fd-boundary-seed}
\end{equation}
Then, for every $0<\beta<1$, the fixed-ray escort states of $(\rho,\sigma)$ agree with those
of $(\widehat\rho_c,\sigma)$, and
\begin{equation}
  D_{\alpha,c\alpha}(\rho\Vert\sigma)
  =\frac{\alpha}{1-\alpha}
   \int_\alpha^1
   \frac{D((\widehat\rho_c)_\beta^{(c)}\Vert\sigma)}{\beta^2}
   \,\dd\beta
   -\frac{\alpha}{1-\alpha}\log q_{1-,c}.
  \label{eq:fd-alpha-c-alpha}
\end{equation}
For $c=1$, $q_{1-,1}=\Tr(P\rho)$. 

Moreover
\begin{equation}\label{eq:qbound}
  0\le q_{1-,c}\le1, 
\end{equation}
with
\begin{equation}\label{eq:endpoint-bounds}
    q_{1-,c}=1\iff s(\rho)\le s(\sigma), \qquad q_{1-,c}=0\iff s(\rho)\perp s(\sigma),
\end{equation}
so the boundary term in \eqref{eq:fd-alpha-c-alpha} is nonnegative and vanishes precisely under support inclusion.

Finally, the boundary correction arises as the limit
\begin{equation}\label{eq:endpointjump}
  \lim_{\beta\uparrow1}Q^{(c)}_{\beta}(\rho\|\sigma)=q_{1-,c}.
\end{equation}
\end{proposition}

\begin{proof}
The powers of $\sigma$ restrict the fixed-ray product to $P\cH$, so
\begin{equation}
  P\rho^{1/c}P=q_{1-,c}^{1/c}\widehat\rho_c^{1/c}.
\end{equation}
Consequently,
\begin{equation}
  Q_{\beta,c\beta}(\rho\|\sigma)
  =q_{1-,c}^{\beta}
   Q_{\beta,c\beta}(\widehat\rho_c\|\sigma),
  \qquad 0<\beta<1.
  \label{eq:fd-boundary-scaling}
\end{equation}
The scalar cancels in the normalized escort.  Taking logarithms at
$\beta=\alpha$ and applying the supported-pair representation proved in
Section~\ref{subsec:matrix-approximation} to the
supported pair $(\widehat\rho_c,\sigma)$ proves
\eqref{eq:fd-alpha-c-alpha}.  
If $q_{1-,c}=0$, the lower-order moment
vanishes, and the convention $\log0=-\infty$ gives the asserted value.

Applying Lemma~\ref{lem:compression} with $B=\rho^{1/c}$, $P=s(\sigma)$ and using $\Tr(B^c)=\Tr(\rho)=1$ proves \eqref{eq:qbound} and \eqref{eq:endpoint-bounds}.

Finally, to prove the endpoint jump \eqref{eq:endpointjump} define
\[
    t_\beta\defeq\frac{1-\beta}{2c \beta}.
\]
As $t_\beta\longrightarrow0$, $\sigma_P^{t_\beta}\longrightarrow P $ in norm. Since, $c\beta \longrightarrow c>0$, joint norm continuity of $(A,t)\longmapsto A^t$ for positive finite dimensional matrices $A$ and exponents bounded away from zero gives
\[
   \left( X_{\beta}^{(c)}\right)^{c\beta}\longrightarrow (P \rho^{1/c} P)^c.
\]
Taking traces proves
\[
    \lim_{\beta\uparrow1} Q_\beta^{(c)}(\rho\|\sigma)= \Tr((P \rho^{1/c} P)^c)=q_{1-,c}. \qedhere
\]
\end{proof}

\section{Stability and the escort average}
\label{app:escort-stability}

Assume the branchwise hypotheses of
Theorem~\ref{thm:fixed-ray-integral}: for \(0<\alpha<1\), assume
\(s(\psi)\leq s(\omega)\), whereas for \(\alpha>1\), assume
\(Q_{\alpha,z}(\psi\|\omega)<\infty\).  Put \(c=z/\alpha\) and write
\(\|\cdot\|_1=\|\cdot\|_{M_*}\) for the predual norm.

Pinsker's inequality for Araki relative entropy \cite[Theorem~3.1]{Hiai:1981sufficiency} gives, pointwise,
\begin{equation}
  R_c(\beta)
  =
  D(\psi_\beta^{(c)}\Vert\omega)
  \geq
  \frac12
  \|\psi_\beta^{(c)}-\omega\|_1^2,
  \qquad \beta\in I_\alpha .
  \label{eq:escort-pinsker-pointwise}
\end{equation}
The result in \cite{Hiai:1981sufficiency} assumes a cyclic and separating vector. This directly covers our setting if $\omega$ is faithful normal. For arbitrary $M$, the inequality follows by restricting to the $\sigma$-finite corner $eMe$ with $e=s(\psi_\beta^{(c)}+\omega)$, where the relative entropy and predual distance are unchanged.

To formulate the averaged inequality without imposing measurability, define the upper Lebesgue integral of a nonnegative function
\(f\) by
\begin{equation}
  \int_{I_\alpha}^{*} f\,\dd\mu_\alpha
  :=
  \inf\left\{
    \int_{I_\alpha}g\,\dd\mu_\alpha:
    g\text{ is Borel measurable and }f\leq g
  \right\}.
  \label{eq:upper-lebesgue-integral}
\end{equation}
The proof of Corollary~\ref{cor:raywise-monotonicity} in Section~\ref{sec:proof-alpha-z} shows that $R_c$ is
Borel measurable.
Using this together with
\(\|\psi_\beta^{(c)}-\omega\|_1^2\leq2R_c(\beta)\),
Theorem~\ref{thm:fixed-ray-integral} implies
\begin{equation}
  D_{\alpha,z}(\psi\Vert\omega)
  \geq
  \frac12
  \int_{I_\alpha}^{*}
    \|\psi_\beta^{(c)}-\omega\|_1^2
    \,\dd\mu_\alpha(\beta).
  \label{eq:escort-pinsker-upper-integral}
\end{equation}
If the scalar function
\[
  \beta\longmapsto
  \|\psi_\beta^{(c)}-\omega\|_1^2
\]
is \(\mu_\alpha\)-measurable, then the upper integral is the ordinary
Lebesgue integral and
\begin{equation}
  D_{\alpha,z}(\psi\Vert\omega)
  \geq
  \frac12
  \int_{I_\alpha}
    \|\psi_\beta^{(c)}-\omega\|_1^2
    \,\dd\mu_\alpha(\beta).
  \label{eq:escort-pinsker-ordinary}
\end{equation}
Finite dimensionality is one sufficient condition for this scalar
measurability, since the escort path is then trace-norm continuous, but it
is not necessary.

For $\varepsilon>0$, let
\begin{equation}
  f(\beta)= \|\psi_\beta^{(c)}-\omega\|_1^2, \quad A_\varepsilon  \defeq \{f\geq \varepsilon^2\},
\end{equation}
and let $\mu_\alpha^*$ be the outer measure induced by $\mu_\alpha$.\footnote{If $\mu$ is a measure on $(X,\Sigma)$, its induced outer measure is the function $\mu^*:2^X\to[0,\infty]$ defined by $\mu^*(A)=\inf\{\mu(B):B\in\Sigma,A\subseteq B\}$, where $2^X$ denotes the collection of all subsets of $X$.}
Then for every Borel $g\geq f$,
\begin{equation}
    \mu_\alpha^*(A_\varepsilon)\leq \mu_\alpha\{g\geq \varepsilon^2\}\leq \frac{1}{\varepsilon^2} \int g \, d\mu_\alpha.
\end{equation}
Taking the infimum over such $g$ and using \eqref{eq:escort-pinsker-upper-integral} gives, for every \(\varepsilon>0\),
\begin{equation}
  \mu_\alpha^*\!\left(
    \left\{\beta\in I_\alpha:
      \|\psi_\beta^{(c)}-\omega\|_1\geq\varepsilon
    \right\}
  \right)
  \leq
  \min\left\{
    1,\frac{2D_{\alpha,z}(\psi\Vert\omega)}{\varepsilon^2}
  \right\},
  \label{eq:escort-concentration}
\end{equation}
with $z=c\alpha$.

\subsection{Finite-dimensional escort average}
Suppose now that \(M=\mathcal B(\mathcal H)\) with
\(\dim\mathcal H<\infty\), and identify \(\psi\) and \(\omega\) with their
density matrices \(\rho\) and \(\sigma\).  Under the hypotheses above, the
map
\[
  \beta\longmapsto\rho_\beta^{(c)}
\]
is continuous in trace norm on \(I_\alpha\), and hence is Bochner integrable
with respect to \(\mu_\alpha\).  Since \(\mu_\alpha\) is a probability
measure, its barycenter
\begin{equation}
  \overline\rho_{\alpha,c}
  :=
  \int_{I_\alpha}\rho_\beta^{(c)}\,
  \dd\mu_\alpha(\beta)
  \label{eq:escort-barycenter}
\end{equation}
is again a density matrix.  The generalized Donald identity for continuous ensembles
\cite[Lemma~4, Eq.~(16)]{Holevo:2004ensembles}, together with
Theorem~\ref{thm:fixed-ray-integral}, gives the exact decomposition
\begin{equation}
  D_{\alpha,z}(\rho\Vert\sigma)
  =
  D(\overline\rho_{\alpha,c}\Vert\sigma)
  +
  \int_{I_\alpha}
    D(\rho_\beta^{(c)}\Vert\overline\rho_{\alpha,c})\,
    \dd\mu_\alpha(\beta).
  \label{eq:donald-holevo-escort}
\end{equation}
This is easily seen by adding and subtracting
$\Tr(\rho_\beta^{(c)}\log\overline\rho_{\alpha,c})$ in the integrand and
using the barycenter definition. If $\overline\rho_{\alpha,c}$ is singular, let $v\in \ker \overline\rho_{\alpha,c}$. Then the continuous non-negative function
\begin{equation}
    \beta\longmapsto \langle{v, \rho_\beta^{(c)} v}\rangle
\end{equation}
has zero $\mu_\alpha$ integral. Since $\mu_\alpha$ has full support, the above function vanishes identically. Hence
\begin{equation}
    s(\rho_\beta^{(c)} ) \leq s(\overline\rho_{\alpha,c}),
\end{equation}
for every $\beta\in I_\alpha$. This support inclusion ensures that the add-and-subtract
calculation involving $\log\overline\rho_{\alpha,c}$ is legitimate, with
the logarithm taken on its support.

The second term in \eqref{eq:donald-holevo-escort} is the Holevo information of the
\(\mu_\alpha\)-weighted escort ensemble.  Denote it by
\begin{equation}
  \chi_{\alpha,c}
  :=
  \int_{I_\alpha}
    D(\rho_\beta^{(c)}
      \Vert\overline\rho_{\alpha,c})\,
    \dd\mu_\alpha(\beta).
  \label{eq:escort-holevo-information}
\end{equation}
Applying Pinsker's inequality pointwise to
\((\rho_\beta^{(c)},\overline\rho_{\alpha,c})\) and integrating gives
\begin{equation}
  \chi_{\alpha,c}
  \geq
  \frac12
  \int_{I_\alpha}
    \bigl\|
      \rho_\beta^{(c)}-\overline\rho_{\alpha,c}
    \bigr\|_1^2\,
    \dd\mu_\alpha(\beta).
  \label{eq:escort-holevo-pinsker}
\end{equation}
Equivalently, the Donald decomposition admits the stability refinement
\begin{equation}
  D_{\alpha,z}(\rho\Vert\sigma)
  \geq
  D(\overline\rho_{\alpha,c}\Vert\sigma)
  +
  \frac12
  \int_{I_\alpha}
    \bigl\|
      \rho_\beta^{(c)}-\overline\rho_{\alpha,c}
    \bigr\|_1^2\,
    \dd\mu_\alpha(\beta).
  \label{eq:donald-holevo-stability}
\end{equation}
Thus the Holevo term quantitatively controls the mean-square trace-norm
dispersion of the escort states around their barycenter.  It vanishes
precisely when
\(\rho_\beta^{(c)}=\overline\rho_{\alpha,c}\) for
\(\mu_\alpha\)-almost every \(\beta\).  By trace-norm continuity and the
full support of \(\mu_\alpha\) on \(I_\alpha\), this is equivalent to the
escort trajectory being constant on \(I_\alpha\).

\subsection{Algebraic rigidity of the escort Holevo term}
The almost-everywhere part of the preceding vanishing criterion does not
depend on finite dimensionality.  Let \(M\) be an arbitrary von Neumann
algebra and suppose that the escort path
\[
  \beta\longmapsto\psi_\beta^{(c)}\in M_*
\]
is strongly measurable as an \(M_*\)-valued map.  Since
\(\|\psi_\beta^{(c)}\|_{M_*}=1\), the path is Bochner integrable, and
\begin{equation}
  \overline\psi_{\alpha,c}
  :=
  \int_{I_\alpha}
    \psi_\beta^{(c)}\,\dd\mu_\alpha(\beta)
  \in M_*
  \label{eq:algebraic-escort-barycenter}
\end{equation}
is a normal state.  Moreover, the map
\[
  \varphi\longmapsto
  D(\varphi\Vert\overline\psi_{\alpha,c})
\]
is lower semicontinuous, and hence Borel measurable, for the
predual-norm topology.  Combined with the assumed strong
measurability of
\(\beta\mapsto\psi_\beta^{(c)}\), this shows that
\[
  \beta\longmapsto
  D\!\left(
    \psi_\beta^{(c)}
    \middle\Vert
    \overline\psi_{\alpha,c}
  \right)
\]
is measurable with respect to the \(\mu_\alpha\)-completion of
\(\mathcal B(I_\alpha)\), the Borel $\sigma$-algebra of $I_\alpha$.  Hence the extended nonnegative quantity
\begin{equation}
  \chi_{\alpha,c}
  :=
  \int_{I_\alpha}
    D\!\left(
      \psi_\beta^{(c)}
      \middle\Vert
      \overline\psi_{\alpha,c}
    \right)
    \,\dd\mu_\alpha(\beta)
  \label{eq:algebraic-escort-holevo}
\end{equation}
is well defined.  Pinsker's inequality for normal states on \(M\) \cite{Hiai:1981sufficiency} gives
\begin{equation}
  \chi_{\alpha,c}
  \geq
  \frac12
  \int_{I_\alpha}
    \bigl\|
      \psi_\beta^{(c)}-\overline\psi_{\alpha,c}
    \bigr\|_{M_*}^{\,2}
    \,\dd\mu_\alpha(\beta).
  \label{eq:algebraic-escort-holevo-pinsker}
\end{equation}
Consequently,
\begin{equation}
    \chi_{\alpha,c}=0
    \quad\Longleftrightarrow\quad
    \psi_\beta^{(c)}
    =
    \overline\psi_{\alpha,c}
    \quad\text{for \(\mu_\alpha\)-almost every
    \(\beta\in I_\alpha\)}.
  \label{eq:algebraic-escort-holevo-rigidity}
\end{equation}
Thus neither finite dimensionality nor continuity in the R\'enyi order is
needed for this almost-everywhere rigidity statement.  For this
equivalence alone, strong measurability may be replaced by the weaker
assumptions that the escort family admits a normal \(M_*\)-valued barycenter
and that the displayed relative-entropy profile is measurable.

\end{document}